\documentclass[runningheads]{llncs}
\usepackage[T1]{fontenc}
\usepackage{graphicx}
\usepackage[T1]{fontenc}
\usepackage{graphicx}
\usepackage{xspace}
\usepackage{tikz}
\usetikzlibrary{arrows, calc, shapes,petri}
\usepackage{pgfplots} 
\pgfplotsset{compat = newest}
\usepackage{paralist}
\usepackage{amssymb}
\usepackage{amsmath}
\usepackage{todonotes}
\usepackage{thmtools, thm-restate}
\usepackage{hyperref}
\usepackage{tabularx}
\usepackage{booktabs}
\usepackage{rotating}
\usepackage{eurosym}

\usepackage{color}

\usepackage{amssymb,bm}
\usepackage{pifont}
\newcommand{\cmark}{\ding{51}}%
\newcommand{\xmark}{\ding{55}}
\newcommand{\smark}{\ensuremath{\bm{\thicksim}}}

\newcolumntype{C}{>{\centering\arraybackslash}X}
\newcolumntype{R}{>{\raggedleft\arraybackslash}X} 
\newcolumntype{L}{>{\raggedright\arraybackslash}X} 
\newcolumntype{P}[2]{%
  >{\begin{turn}{#1}\begin{minipage}{#2}\raggedright}l%
  <{\end{minipage}\end{turn}}%
}
\newcolumntype{u}[1]{%
  >{\begin{minipage}{#1}\raggedleft}l%
  <{\end{minipage}}%
}

\newcommand{\mc}[1]{\mathcal{#1}}
\newcommand{\m}[1]{\mathsf{#1}}

\newcommand{\tup}[1]{\langle{#1}\rangle}
\newcommand{\events}{E}
\newcommand{\objects}{\mathcal O}
\newcommand{\otypes}{\Sigma_{obj}}
\newcommand{\ltypes}{\Sigma_{list}}
\newcommand{\activities}{\mc A}
\newcommand{\timestamps}{\mathbb T}
\newcommand{\proj}{\pi}
\newcommand{\projact}{\proj_{\mathit{act}}}
\newcommand{\projobj}{\proj_{\mathit{obj}}}
\newcommand{\projtime}{\proj_{\mathit{time}}}
\newcommand{\projtrace}{\proj_{\mathit{trace}}}
\newcommand{\projvalue}{\proj_{\mathit{val}}}
\newcommand{\otype}{\vartype}

\newcommand{\pre}[1]{\bullet{#1}}
\newcommand{\post}[1]{{#1}\bullet}
\newcommand{\logtrace}{\mathbf e}
\newcommand{\goto}[1]{\mathrel{\raisebox{-2pt}{$\xrightarrow{#1}$}}}

\newcommand{\co}[1]{\mathtt{#1}} 

\newcommand{\GG}{\mc G} 
\newcommand{\LL}{\mc L} 
\newcommand{\NN}{\mc N} 

\newcommand{\exaref}[1]{Ex.~\ref{exa:#1}}
\newcommand{\defref}[1]{Def.~\ref{def:#1}}

\newcommand{\lemref}[1]{Lem.~\ref{lem:#1}}
\newcommand{\figref}[1]{Fig.~\ref{fig:#1}}
\newcommand{\tabref}[1]{Tab.~\ref{tab:#1}}

\tikzstyle{place}=[draw, circle, inner sep=1.5pt, line width=.7pt, scale=.8, minimum width=6mm]
\tikzstyle{trans}=[draw, rectangle, inner sep=1.5pt, line width=.7pt, scale=.8, minimum width=6mm, minimum height=6mm, fill=gray!10]
\tikzstyle{arc}=[draw, ->, line width=.5pt]
\tikzstyle{fatarc}=[draw, ->, line width=.5pt, double]
\tikzstyle{action}=[scale=.6]
\tikzstyle{starttoken}=[regular polygon, regular polygon sides=3,minimum width=2mm,fill=black,inner sep=0pt,rotate=30]
\tikzstyle{endtoken}=[regular polygon, regular polygon sides=4,minimum width=2mm,fill=black,inner sep=0pt]
\colorlet{ocolor}{blue!70!green!20}
\colorlet{odcolor}{blue!80!green!40}
\colorlet{oscolor}{magenta!90!red!20}
\colorlet{icolor}{yellow!90!red!20}
\colorlet{oicolor}{green!90!blue!20}
\tikzstyle{oplace}=[place,fill=ocolor]
\tikzstyle{odplace}=[place,fill=odcolor]
\tikzstyle{osplace}=[place,fill=oscolor]
\tikzstyle{otrans}=[trans,fill=ocolor]
\tikzstyle{iplace}=[place,fill=icolor]
\tikzstyle{itrans}=[trans,fill=icolor]
\tikzstyle{oiplace}=[place,fill=oicolor]
\tikzstyle{idplace}=[place,fill=red!30] 
\tikzstyle{dplace}=[place,fill=magenta!20]             
\tikzstyle{oidplace}=[place,fill=black!20]
\tikzstyle{oitrans}=[trans,fill=oicolor]
\tikzstyle{mixtrans}=[trans,shading = axis,rectangle, left color=ocolor, right color=icolor,shading angle=135]
\tikzstyle{oimixtrans}=[trans,shading = axis,rectangle, left color=ocolor, right color=oicolor,shading angle=135]
\tikzstyle{insc}=[scale=.8]
\tikzstyle{splitme}=[rectangle split, rectangle split horizontal,rectangle split parts=2]
\tikzstyle{splitmemore}=[rectangle split, rectangle split horizontal,rectangle split parts=3]
\tikzstyle{log}=[fill=white]
\tikzstyle{model}=[fill=gray!10]

\newcommand{\objtuples}{\vec {\mathcal{TOK}}}

\newcommand{\dom}{\mathit{dom}}

\newcommand{\vartype}{\mathit{type}}

\newcommand{\eventlog}{L}

\renewcommand{\mod}{{\mathit{mod}}}

\newcommand{\cost}{{\mathit{cost}}}
\newcommand{\run}{\rho}

\newcommand{\moves}{\mathit{moves}}
\newcommand{\SKIP}{{\gg}}
\newcommand{\restrlog}{|_{\mathit{log}}}
\newcommand{\restrmod}{|_{\mathit{mod}}}

\newcommand{\nuvarset}{\Upsilon}
\newcommand{\listvarset}{\varset_{\mathit{list}}}

\newcommand{\setsym}[1]{\mathit{#1}}
\newcommand{\places}{P}
\newcommand{\inflow}{F_{in}}
\newcommand{\outflow}{F_{out}}

\newcommand{\transitions}{T}

\newcommand{\varset}{\mathcal{V}}

\newcommand{\funsym}[1]{\mathit{#1}}
\newcommand{\guass}{\funsym{guard}}
\newcommand{\coloring}{\funsym{color}}
\newcommand{\set}[1]{\{#1\}}
\newcommand{\allvars}{\mathcal{X}}

\newcommand{\invars}[1]{\setsym{\setsym{vars}_{in}}(#1)}
\newcommand{\outvars}[1]{\setsym{\setsym{vars}_{out}}(#1)}

\newcommand{\vars}[1]{\setsym{vars}(#1)}
\newcommand{\colors}{Col}
\newcommand{\Dom}{dom}
\newcommand{\marking}{M}

\newcommand{\listtype}[1]{[{#1}]}
\newcommand{\tracenet}{T}

\newcommand{\eid}[1]{\mathtt{\#}_{#1}}

\newcommand{\mynet}{DOPID\xspace}
\newcommand{\mynets}{DOPIDs\xspace}
\newcommand{\anet}{a \mynet}

\newcommand{\eqn}{\,{=}\,}

\newcommand{\rational}{\mathtt{rat}}
\newcommand{\bool}{\mathtt{bool}}
\newcommand{\integer}{\mathtt{int}}
\newcommand{\str}{\mathtt{string}}
\newcommand{\constraint}{c}
\newcommand{\inn}{\,{\in}\,} 
\newcommand{\fset}{\mathtt{finset}}

\newcommand{\varsin}[1]{\setsym{vars}(#1)}

\newcommand{\assignments}[1][\varset]{\mathit{Assign}^{#1}}

\newcommand{\extended}[1]{#1}

\begin{document}
\title{Object-centric Processes with\\ Structured Data and Exact Synchronization\\  \extended{(Extended Version)}}
\subtitle{Formal Modelling and Conformance Checking}
\titlerunning{Object-centric processes with structured data and exact synchronization}
%
\author{Alessandro Gianola\inst{1}
\and
Marco Montali\inst{2}
\and
Sarah Winkler\inst{2}
}
\authorrunning{A. Gianola et al.}
%
\institute{INESC-ID/Instituto Superior Técnico, Universidade de Lisboa, Portugal 
\email{alessandro.gianola@tecnico.ulisboa.pt}\\
\and
Free University of Bozen-Bolzano, Italy\\
\email{\{montali,winkler\}@inf.unibz.it}}
\maketitle              
\begin{abstract}
Real-world processes often involve interdependent objects
that also carry
data values, such as integers, reals, or strings. 
However, existing process formalisms fall short to combine key modeling features, such as tracking object identities, supporting complex datatypes,  handling dependencies among them, and object-aware synchronization.
Object-centric Petri nets with identifiers (OPIDs) partially address these needs but treat objects as unstructured identifiers (e.g., order and item IDs), overlooking the rich semantics of complex data values (e.g., item prices or other attributes).
To overcome these limitations, we introduce data-aware OPIDs (DOPIDs), a framework that strictly extends OPIDs by incorporating structured data manipulation capabilities, and full synchronization mechanisms. 
In spite of the expressiveness of the model, we show that it can be made operational: 
Specifically, we define 
a novel conformance checking approach leveraging satisfiability modulo theories (SMT) to compute data-aware object-centric alignments. 

\end{abstract}

\section{Introduction}
\label{sec:intro}

In recent years, research in 
information systems supporting the execution of business and work processes has increasingly stressed the need for multi-perspective process models. Prominently, the goal has been to tackle intricate links between the control flow
of a process, and the data 
with which the control flow interacts. 

Several new process modelling paradigms and corresponding languages consequently emerged, ranging from case-handling \cite{AaWG05,HewW16} to artifact-centric \cite{CohH09,HBGV13} and object-aware approaches \cite{KunzleR11,SVDD23,Breitmayer0PR24}. In parallel,
%
a growing stream of research has spawned different formal models of data-aware processes with the twofold aim to support \emph{representation} and \emph{computation}, on the one hand covering relevant mo\-del\-ling constructs, on the other 
providing
effective algorithmic techniques for process analysis and mining, such as  
automated discovery and conformance checking. 
Within this stream, representations typically rely on data-aware extensions of Petri nets, the most widely employed formalism for describing, analysing, and mining processes, with two emerging directions.

The first focuses on enriching case-centric processes by incorporating structured case attributes (e.g., the price of a product, the age and name of a customer, as well as more complex structures such as a persistent relational storage). This supports expressing how activities in the process read and write these variables, and how decision points use these variables to express routing conditions for cases. A prime example in this vein is that of Data Petri nets (DPNs \cite{MannhardtLRA16,LeoniFM18}).

The second direction 
aims instead at lifting the case-centricity assumption, tackling so-called \emph{object-centric processes} where multiple objects, interconnected via complex one-to-many and many-to-many relationships, are co-evolved by the process (e.g., orders containing multiple products, shipped in packages that may mix up products from different orders). It has been 
pointed out that straight-jacketing this complexity through a single case notion yields misleading process analysis and mining results \cite{Aalst19,BeMA23}. Several proposals have been brought forward in this direction, such as object-centric Petri nets \cite{AalstB20}, synchronous proclets \cite{Fahland19}, and Petri nets with identifiers (PNIDs) \cite{PWOB19,WerfRPM22,GianolaMW24}, possibly equipped with an external relational storage \cite{MonR17,GhilardiGMR22}.
To a varying extent, they cover essential modelling features like:
\begin{inparaenum}[\it (i)]
\item tracking objects and their possibly concurrent flows, enabling the independent progression of objects (e.g. package shipments and order notifications);
\item object creation and manipulation, handling dependencies such as one-to-one and one-to-many relationships (e.g., adding items to an order or splitting it into multiple packages);
 \item object-aware full synchronization, creating flow dependencies on objects, where an object can flow through an activity only if some (\emph{subset synchronization}) or all (\emph{exact synchronization}) related objects simultaneously flow through that activity; this ensures that an object can proceed through an activity only when certain conditions are met (e.g., initiating order billing only when some or all associated packages have been delivered).
\end{inparaenum}

Two main open challenges emerge in 
the tradeoff of these two directions.
First, there is a lack of object-centric models that offer both object relationship manipulation and corresponding synchronization mechanisms: while approaches based on PNIDs provide fine-grained constructs for handling objects and their mutual relationships, they fall short in fully addressing synchronization in its different flavors; complementarily, alternative approaches in the object-centric spectrum suffer from under-specification issues \cite{Aals23,GianolaMW24}. Second, object-centric models completely lack support for attributes and corresponding data conditions.

\emph{The ultimate goal of this work is to address these two research questions through a unified formalism supporting at once fine-grained modelling features for objects, relationships, attributes, and complex data conditions to express transition guards, subset/exact synchronization, and combinations thereof}.
The following example inspired by \cite{GianolaMW24,AalstB20} motivates the need for such a formalism:

\begin{example}
\label{exa:intro}
In an order-to-shipment process, executions involve the following activities:
\begin{inparaenum}[\it (i)]
\item $\m{place\:order}$ creates a new order with an arbitrary number of products, and the customer can indicate the number of days expected for delivery;.
\item Payment can be done via credit card or bank transfer, reflected by $\m{pay\:cc}$ and $\m{pay\:bt}$, respectively. However, the former is only applicable if the total cost of the order is below 1000\texteuro.
\item $\m{pick\:item}$ fetches a product from the warehouse;
\item $\m{ship}$ collects an order with all its products for shipment. It also determines one of two shipment modes, namely $\m{car}$ or $\m{truck}$, depending on whether the number of days for delivery is below or above 5.
\end{inparaenum}

This example requires a variety of sophisticated modelling features:
during order creation, an arbitrary number of product objects can be included (in the sequel we call this \emph{multi-object transfer}); items are associated with an order (\emph{object relations}), and \emph{all} items of an order must be included in shipment (in the sequel called \emph{exact synchronization}, as opposed to \emph{subset synchronization}); one needs to reason over arithmetic data like the cost of a product and the total cost of an order, and transition guards are needed (\emph{structured data support}).
\end{example}

Specifically,  we provide a twofold contribution in \emph{representation} and \emph{computation}, to handle such processes. As for \emph{representation}, we start from  OPIDs \cite{GianolaMW24}, the most sophisticated PNID-based formalism: they support all main object-centric modelling features except exact synchronization. We lift OPIDs into a new class of PNIDs called \mynets, which at once close the gap regarding synchronization, and 
add rich data support for
a variety of data types, together with conditions expressed over such data that can involve arithmetic, uninterpreted functions, object properties, and  advanced forms of aggregation. Aggregation emerges as a natural, non-trivial new modelling construct arising from the interplay between data conditions \`a la DPNs, and the fact that they are now applied over the attributes of possibly multiple objects at once. As for \emph{computation}, we consider conformance checking, and show that
existing SMT-based techniques 
can be lifted to fully cover all features of DOPIDs.
We do so by integrating and extending the SMT encodings for conformance checking separately studied for OPIDs \cite{GianolaMW24} and DPNs \cite{FelliGMRW23,FelliGMRW23b}.
Finally, we provide a novel proof-of-concept implementation of our approach to witness its feasibility.

\section{Related Work and Modelling Features}
\label{sec:related}

To highlight key modeling features, 
we reviewed literature on Petri nets enriched with case attributes \cite{FelliLM21,CalvaneseGGMR19,Gianola23} and object-centric features \cite{Aalst19,BeMA23,Aalst23}. 
\tabref{features} shows a summary of these features and their implementation in various approaches. At the end of Section~\ref{sec:nets} we discuss the expressivity of DOPIDs more generally.

\newcommand{\header}[1]{
  \multicolumn{1}{P{90}{1.9cm}}{#1}
}

\newcommand{\implicit}{\textsf{imp.}}
\newcommand{\explicit}{\textsf{exp.}}

\begin{table}[t]
\centering
\resizebox{.8\textwidth}{!}{
\begin{tabularx}{\textwidth}{
m{1.6cm}CCCCCCCCCCCCCCCCCCCCCCCCCc}
&
&
\header{object \\[-.7ex]creation}
&
&
\header{object \\[-.7ex]removal}
&
&
\header{concurrent\\[-.7ex] object flows}
&
&
\header{multi-object\\[-.7ex] transfer}
&
&
\header{multi-object\\[-.7ex] spawning}
&
&
\header{object \\[-.7ex]relations}
&
&
\header{subset\\[-.7ex] sync}
&
&
\header{exact\\[-.7ex] sync}
&
&
\header{coreference}
&
&
\header{struct. data\\ ~}
&
&
\header{object \\[-.7ex] reference}
&
&
\header{conformance}
&
&
\\
\toprule

OC nets \cite{AalstB20}
&&
\cmark
&&
\cmark
&&
\cmark
&&
\cmark
&&
\cmark
&&
\xmark
&&
\xmark
&&
\xmark
&&
\xmark
&&
\xmark
&&
\implicit
&&
\cite{LissAA23}
\\
\midrule
synchronous proclets \cite{Fahland19}
&&
\cmark
&&
\cmark
&&
\cmark
&&
\smark
&&
\cmark
&&
\cmark
&&
\cmark
&&
\cmark
&&
\xmark
&&
\xmark
&&
\implicit
&&
\xmark
\\
\midrule
DPNs \cite{LeoniFM18}
&&
\cmark
&&
\cmark
&&
\cmark
&&
\xmark
&&
\xmark
&&
\xmark
&&
\xmark
&&
\xmark
&&
\xmark
&&
\cmark
&&
\implicit
&&
\cite{MannhardtLRA16}
\\
\midrule
PNIDs \cite{PWOB19,GhilardiGMR22,WerfRPM22}
&&
\cmark
&&
\cmark
&&
\cmark
&&
\smark
&&
\smark
&&
\cmark
&&
\smark
&&
\xmark
&&
\cmark
&&
\xmark
&&
\explicit
&&
\xmark
\\
\midrule
OPIDs \cite{GianolaMW24}
&&
\cmark
&&
\cmark
&&
\cmark
&&
\cmark
&&
\cmark
&&
\cmark
&&
\cmark
&&
\xmark
&&
\cmark
&&
\xmark
&&
\explicit
&&
\cite{GianolaMW24}
\\
\midrule
DABs \cite{CalvaneseGGMR19}
&&
\cmark
&&
\cmark
&&
\xmark
&&
\xmark
&&
\xmark
&&
\xmark
&&
\xmark
&&
\xmark
&&
\xmark
&&
\cmark
&&
no
&&
\xmark
\\
\midrule
\textbf{\mynets}
&&
\cmark
&&
\cmark
&&
\cmark
&&
\cmark
&&
\cmark
&&
\cmark
&&
\cmark
&&
\cmark
&&
\cmark
&&
\cmark
&&
\explicit
&&
\emph{here}
\\
\bottomrule
\end{tabularx}
}
\caption{Comparison of Petri net-based object-centric process modelling languages along main modelling features, tracking which approaches support conformance. {\cmark} indicates full, direct support, {\xmark} no support, and {\smark} indirect or partial support.}
\label{tab:features}
\vspace*{-7mm}
\end{table}

The first crucial feature is the incorporation of constructs for \emph{creating and deleting objects}. Different approaches vary based on whether objects are \emph{explicitly referenced} within the model or are only \emph{implicitly manipulated}. Another critical aspect is the ability of objects to \emph{flow concurrently} and independently; for example, items can be picked while their corresponding order is paid (cf. \emph{divergence} in \cite{Aalst19}). Additionally, models may support the \emph{simultaneous transfer of multiple objects} of the same type, such as processing several items in a single transaction. Transitions in these models must account for the manipulation of multiple objects, either of the same type or different types, at the same time (cf. \emph{convergence} in \cite{Aalst19}).
A type of convergence occurs when a single transition, given a parent object, \emph{generates an unbounded number of child objects} 
that
are all linked to the parent; e.g., placing an order can create unboundedly many associated items. Once this parent-child \emph{one-to-many relationship} is established, other forms of convergence, such as synchronizing transitions, can be introduced. These transitions allow a parent object to evolve  only if \emph{some} (\emph{subset synchronization}) or \emph{all} its child objects (\emph{exact synchronization})  are in a certain state. In addition, advanced \emph{coreference} techniques can be employed to simultaneously examine and evolve multiple interconnected objects. Finally, an essential feature is the support for advanced and \emph{structured data types}, such as integers, reals, lists, and arrays. They enhance the 
objects 
manipulated by the process with additional information, allowing users to define complex constraints that act as guards for the transitions in the process model, possibly 
incorporating
background knowledge.
This final feature, unlike the others, is more characteristic of data-aware extensions of \emph{case-centric} process models \cite{MontaliC16}, where the focus is traditionally placed on the complex structure of data values, often governed by relational logical theories \cite{DamaggioDHV11}. Among these approaches, the most advanced framework is the DAB model \cite{CalvaneseGGMR19,GhilardiGMR21}, which supports rich forms of database-driven data and sophisticated forms of reasoning.
In process mining, multi-perspective models capable of incorporating richer data representations while expressing concurrent flows have also been introduced. A prominent example are Data Petri Nets (DPNs) \cite{FelliLM21}, a Petri net-based formalism that, while more expressive, remains case-centric.

Regarding object-centric models, several Petri net-based formal models have been introduced \cite{SoSD22,AalstB20,Fahland19,PWOB19,GhilardiGMR22,WerfRPM22}.
Resource-constrained $\nu$-Petri nets \cite{SoSD22} constitute the first formal model supporting a basic form of object-centricity, but without relationships between objects.
Object-centric nets \cite{AalstB20} offer an implicit approach to object centricity where places and transitions
have
different object types. Simple arcs match with a single object at a time, while 
variable
arcs handle arbitrarily many objects of the same type. However, the lack of 
object relationships prevents modeling object synchronization and coreference. Alignment-based conformance checking for this approach is developed in \cite{LissAA23}.
Synchronous proclets \cite{Fahland19} offer a framework that can implicitly express the tracking of objects and their mutual relationships. They include specialized constructs to support the described types of convergence, including subset and exact synchronization, though other forms of coreference are not supported. Multi-object transfers are approximated through iteration, processing objects one by one. Conformance checking for proclets is implicitly tackled here for the first time, considering that \mynets generalize proclets.
Petri nets with identifiers (PNIDs), and variants, have been studied in \cite{PWOB19,GhilardiGMR22,WerfRPM22}, though without addressing conformance checking. PNIDs build upon $\nu$-Petri nets by explicitly managing objects and their relationships through identifier tuples. Unlike object-centric nets and proclets, PNIDs lack constructs for manipulating unboundedly many objects in a single transition. As demonstrated in \cite{GhilardiGMR22}, multi-object transfers, spawning, and subset synchronization can be achieved through object coreference and iterative operations. 
In OPIDS \cite{GianolaMW24}, multi-object transfer is natively supported, but no data and only subset synchronization; though during conformance checking, exact synchronization can be obtained as a by-product.
Indeed, no variant of PNID supports data-aware wholeplace operations as necessitated by exact synchronization.
\mynets strictly subsume OPIDS, extending them with 
full object-aware synchronization and
rich data support.
This includes data of numeric, string, or free data types, and complex transition guards involving arithmetic and string operations as well as uninterpreted functions.
These features significantly boost expressivity of the process model.

\section{Object-Centric Event Logs with Data Attributes}
\label{sec:background}

We start from object-centric event logs as in \cite{AalstB20,GianolaMW24}, and enrich them with data attributes. To this end, we assume that data types are divided in two classes: a class $\Sigma_{obj}$ of \emph{object id} types, and a class $\Sigma_{val}$ of \emph{data-value} types.

Consistently with the literature, every object id type $\sigma \in \otypes$ has an (uninterpreted) domain $\Dom(\sigma) \subseteq \objects$, given by all object ids in $\objects$ of type $\sigma$. Such identifiers are used to refer to objects in the real world, and can be compared only for equality and inequality. Examples are order and product identifiers.

To capture the data attributes attached to objects and recorded in event logs, such as for example the price of a product and the delivery address of an order,
we also introduce \emph{data-value domains} for data-value types in $\Sigma_{val}$:
$\Dom(\bool) = \mathbb B$, the booleans;
$\Dom(\integer) = \mathbb Z$, the integers;
$\Dom(\rational) = \mathbb Q$, the rational numbers; and 
$\Dom(\str) = \mathbb S$, the strings over some fixed alphabet; unconstrained finite sets $\Dom(\fset_i)={\mathbb D}_i$, for some finite set ${\mathbb D}_i$, 
We will also support function and relation symbols over all these types to capture implicit properties of objects that are not explicitly manipulated by the process.
E.g., in an order management process where every order has a delivery address (explicitly manipulated by the process) and an owner (implicitly assumed, but not directly manipulated), one may opt for modelling the delivery address as a string, and the owner as an uninterpreted function taking an order id as its only argument. 

In addition to the sets $\otypes$ and $\Sigma_{val}$ for object and data-value types, we fix a set $\activities$ of activities and a set $\timestamps$ of timestamps equipped 
with a total order $<$.
We also consider (partial) assignments from a set of variables $\varset$ to elements of their domain. The set of all such assignments is denoted $\assignments[\varset]$.

\begin{definition}[Event log]
An  \emph{event log (with objects and attributes)} is a tuple
$\eventlog = \tup{\events, \objects, \projact, \projobj,\projtime, \projvalue}$
where:
\begin{inparaenum}[$(i)$]
\item $\events$ is a set of event identifiers;
\item $\objects$ is a set of object identifiers that are typed by a function $\otype\colon \objects \to \otypes$;
\item the functions $\projact\colon E \to \activities$, $\projobj\colon E \to \mc P(\objects)$, and $\projtime\colon E \to \timestamps$ associate
each event $e\in\events$ with an activity, a set of affected objects, and a timestamp, respectively, such that
for every $o\in \objects$ the timestamps $\projtime(e)$ of all events $e$ such that $o \in \projobj(e)$ are all different;
\item the function $\projvalue\colon E \to 
\assignments[\varset]
$  associates each event $e\in\events$ with a set of data-values for attributes in $\varset$.
\end{inparaenum}
\end{definition}

For an event log and an object $o \in \objects$, we write $\projtrace(o)$
for the tuple of events involving $o$, ordered by timestamps. Formally, $\projtrace(o) =\tup{e_1, \dots, e_n}$ such that $\{e_1, \dots, e_n\}$ is the set 
of events in $E$ with $o\in \projobj(e)$, and $\projtime(e_1) < \dots <\projtime(e_n)$.
%
In examples, we often leave $\objects$, $\activities$ and domains
implicit and present an event log $\eventlog$ as a set of tuples 
$\tup{e,\projact(e), \projobj(e), \projtime(e),\projvalue(e)}$ 
representing events. Timestamps are shown as natural numbers, and concrete event ids as $\eid{0}, \eid{1}, \dots$.
The next example demonstrates an event log related to the process outlined in Ex.~\ref{exa:intro}.

\begin{example}
\label{exa:log}
Consider the set of objects $\objects=\{\co{o}_1, \co{p}_1, \co{p}_2\}$ with
$\otype(\co{o}_1)=\mathit{order}$ and $\otype(\co{p}_1) = \otype(\co{p}_2) =$ $\mathit{product}$. 
The event log $\events = \{\eid{0}, \eid{1}, \eid{2}, \eid{3}\}$ with the  events detailed below reports that order $\co{o}_1$ is placed with two products $\co{p}_1,\co{p}_2$ and $3$ days for delivery. Then $\co{o}_1$ is paid by credit card, $\co{p}_1$ is picked, and finally $\co{o}_1$ is shipped only with $\co{p}_1$, confirming the $3$ days and selecting $\m{truck}$ mode:\\[-7pt]
\begin{footnotesize}
\[
\begin{array}{l}
\tup{\eid{0}, \m{place\ order}, \{\co{o}_1,\co{p}_1,\co{p}_2\},\{d \mapsto 3\}, 1}, 
\tup{\eid{1}, \m{pay\ cc}, \{\co{o}_1\}, \emptyset, 4},\notag \\ 
\tup{\eid{2}, \m{pick\ item}, \{\co{o}_1,\co{p}_1\}, \emptyset, 5}, 
\tup{\eid{3}, \m{ship}, \{\co{o}_1,\co{p}_1\}, \{d \mapsto 3,s \mapsto \textsf{truck}\}, 9}
\end{array}
\]
\end{footnotesize}
\end{example}

The notions of object and trace graphs from \cite{AalstB20,GianolaMW24} remain identical also in our setting, but we report them here for completeness.
The \emph{object graph} $\GG_\eventlog$ of an event log $\eventlog$ is the undirected graph with node set $\objects$, and an edge from $o$ to $o'$ if there is some event $e\in E$ such that $o\in \projobj(e)$ and
$o'\in \projobj(e)$. Thus, the object graph indicates which objects share events.
%
Let $X$ be a connected component in $\GG_\eventlog$. Then,
the \emph{trace graph} induced by $X$ is the directed graph $\tracenet_X = \tup{E_X, D_X}$ where:
\begin{inparaenum}
\item[(i)]
the set of nodes $E_X$ is the set of all events $e\in E$ that involve objects in $X$, i.e., such that $X \cap \projobj(e) \neq \emptyset$, and 
\item[(ii)] the set of edges $D_X$ consists of all $\tup{e,e'}$ such that for some 
$o\in \projobj(e) \cap \projobj(e')$, it is $\projtrace(o)=\tup{e_1,\dots,e_n}$
and $e=e_i$, $e'=e_{i+1}$ for some $0\,{\leq}\,i\,{<}\,n$. 
\end{inparaenum}
%
Notably, the notion of trace graph is not modified despite the presence of data: edges only relate events that share objects, independent of possibly shared data values.

\section{Data-Aware Object-Centric Petri Nets with Identifiers}
\label{sec:nets}

We define \emph{data-aware object-centric Petri nets with identifiers} (\mynets), enriching OPIDs \cite{GianolaMW24} with complex data and full synchronization capabilities.

As in PNIDs and OPIDs, objects can be created in \mynets using $\nu$ variables. However, in DOPIDs tokens can carry object ids, data values, or tuples combining these.
The latter account at once for relationships among objects, and attributes connecting objects to data values. So objects can be linked to other objects or data values, e.g.  as in \exaref{order}, where a product tracks the order it belongs to, and an order is associated to its shipment mode. Arcs are labeled with (tuples of) variables to match with objects and relations, as explained later. 
%
In the style of object-centric nets~\cite{AalstB20} and synchronous proclets~\cite{Fahland19}, DOPIDs can spawn and transfer multiple objects at once, using an extension of
the mechanism in OPIDs based on special ``list variables'' that match with \emph{lists of objects}. The refinement consists in the possibility of indicating, when operating over a number of objects by consuming multiple tokens at once, whether one wants to consume \emph{some} or \emph{all} matching objects. The latter case could not be tackled in OPIDs, and is essential to cover exact synchronization.


\smallskip
\noindent
\textbf{Formal Definition.}
Let $\Sigma = \otypes \uplus \Sigma_{val}$ be the set of base types, including object types and data-value types. 
As in colored Petri nets, each place has a \emph{color}: a cartesian product of data types from $\Sigma$. 
More precisely, the set of colors $\colors$ is the set of all $\sigma_1 \times \cdots \times \sigma_m$ such that $m \geq 1$ and $\sigma_i \in \Sigma$ for all $1\leq i\leq m$.
In addition, let the set of list types $\ltypes$ consist of all $[\sigma]$ such that $\sigma \in \otypes$.

In \mynets, tokens are tuples of object ids and data values, each associated with a color.
E.g., to model the process in \exaref{intro}, we want to use $\tup{\co{o}_1}$, $\tup{\co{o}_1, \co{p}_1}$, and $\tup{\co{o}_1, 3}$ as tokens -- respectively representing the order $\co{o}_1$, the relationship indicating that product $\co{p}_1$ is contained in order $\co{o}_1$, and the fact that $\co{o}_1$ has $3$ days of desired delivery by the customer.
In contrast to standard Petri nets which have only indistinguishable black tokens, DOPIDs keep track of object identity when firing transitions. To this end, we use a set of variables $\allvars$ in arc labels, to act as placeholders for objects or lists of objects.
The variables in $\allvars$ are assumed to be
\emph{typed} in the sense that there is a function $\vartype\colon \allvars \to \Sigma \cup \ltypes$ assigning a type to each variable.

\smallskip
\noindent
The set of variables $\allvars = \varset \uplus \varset_{list} \uplus \varset_{list}^{\subseteq} \uplus \varset_{list}^{=} \uplus \nuvarset$ is the disjoint union of: 
\begin{compactenum}
    \item a set $\varset$ of ``normal'' variables that refer to single objects or data values, denoted by lower-case letters like $v$, with a type $\vartype(v) \in \Sigma$;
    \item a set $\varset_{list}$ of list variables referring to a list of objects of the same type, denoted by upper case letters like $U$, with $\vartype(U)=\listtype{\sigma} \in \ltypes$;
    \item  two sets $\varset_{list}^{\subseteq}$ and $\varset_{list}^{=}$ that contain \emph{annotated} list variables $U^{\subseteq}$ and $U^=$ resp., for each list variable $U$ in $\varset_{list}$ -- these will be used to express whether \emph{some} or \emph{all} objects matching the variable must be considered;
    \item a set $\nuvarset$ of variables $\nu$ referring to fresh objects, 
 with $\vartype(\nu) \in \otypes$. 
\end{compactenum} 
 
We assume that infinitely many variables of each kind exist, and for every $\nu \,{\in}\,\nuvarset$, that $\Dom({\vartype(\nu)})$ is 
infinite, for unbounded supply of fresh objects~\cite{RVFE10}.

To capture relationships between objects in consumed and produced tokens when firing transitions,  we need
arc \emph{inscriptions}, which are tuples of variables.

\begin{definition}[Inscription]
\label{def:inscription}
An \emph{inscription} is a tuple $\vec v = \tup{v_1, \dots, v_m}$
such that $m \geq 1$ and $v_i\in \allvars$ for all $i$, but at most one $v_i \in \varset_{list}\uplus \varset_{list}^{\subseteq} \uplus \varset_{list}^{=}$ for $1 \leq i \leq m$.
We call $\vec v$ a 
\emph{transfer-template inscription} if $v_i\in \listvarset$, 
\emph{$\subseteq$-template inscription} if $v_i\in \listvarset^{\subseteq}$, 
or \emph{$=$-template inscription} if $v_i\in \listvarset^{=}$ for some $i$, and a
\emph{simple inscription} otherwise.
\end{definition}

For instance, for $o,p\in\varset$ and $P,Q\in \listvarset$, 
$\tup{o,P}$ is a transfer inscription and $\tup{p}$ a simple one. The inscription $\tup{o,P^=}$ is a $=$-template inscription as it contains the variable $P^=$ in $\listvarset^{=}$, but $\tup{P,P}$ and $\tup{P,Q}$ are not valid inscriptions as they have two list variables.
Note that a variable like $P$ is only a placeholder for a list, it will be instantiated during execution by a concrete list, e.g. $[\co{p}_1, \co{p}_2,\co{p}_3]$. 
By allowing at most one list variable in inscriptions, we restrict to many-to-one relationships between objects. However, recall that many-to-many relationships can be modeled as many-to-one with auxiliary objects, through reification.

Template inscriptions will 
be used to capture an arbitrary number of tokens of the same color: e.g.,
if $o$ is of type \emph{order} and $P$ of type $\listtype{product}$, then
$\tup{o,P}$ refers to a single order with an arbitrary number of products. 
As we will see later, simple, $\subseteq$- and $=$-template inscriptions will be used when consuming tokens, while simple and transfer-template inscriptions will be employed when producing tokens. Specifically, when consuming tokens e.g. carrying order-product pairs from a place $q$, $\tup{o,P^\subseteq}$ selects \emph{some} tokens 
from $q$,
while $\tup{o,P^=}$ selects \emph{all} of them. 
This is essential to model subset and exact synchronization (cf.~Section~\ref{sec:related}). 

We define the \emph{color} of an inscription $\iota =\tup{v_1, \dots, v_m}$ as the tuple of the types of the involved variables, i.e., $\coloring(\iota) =\tup{\sigma_1, \dots, \sigma_m}$ where $\sigma_i=\vartype(v_i)$ if $v_i\in \varset\cup \nuvarset$, and $\sigma_i=\sigma'$ if $v_i$ is a list variable of type $\listtype{\sigma'}$.
Moreover, we set  $\vars{\iota} = \{v_1, \dots, v_m\}$.
E.g. for $\iota=\tup{o,P}$ with $o$, $P$ as above, we have $\coloring(\iota) =$ $\tup{\mathit{order}, \mathit{product}}$ and $\vars{\iota}=\{o,P\}$. 
The set of all inscriptions is denoted $\Omega$.

To 
define guards on transitions, we consider the following definition of \emph{constraints}, where we assume that uninterpreted functions and relations are defined over $\Sigma$ (i.e., all object id and data value domains):

\begin{definition}[Constraints]
\label{def:constraints}
For a set of variables $\varset$ with list variables $\listvarset \subseteq \varset$, a 
\emph{constraint} $\constraint$ and \emph{expressions}
$s$, $n$, $r$, $d$, $k$, $t_D$, and $t_K$ are defined as follows:
 \begin{align*}
 c &::= v_b \mid b  \mid d = d \mid k \geq k \mid k > k\mid  R(d,\dots,d) \mid R(k,\dots, k)\mid c \wedge c \mid \neg c \\
 n &::= v_n \mid z \mid sum(Z) \mid min(Z) \mid max(Z) \mid n + n \mid - n \\
 r &::= v_r \mid q \mid sum(Q) \mid min(Q) \mid max(Q) \mid mean(Q)  \mid r + r \mid - r \\ 
   s &::=\parbox{5cm}{$v_s \mid h  \mid f(s,\dots,s)$}
 d ::=   s  \mid 
     f_w(d,\dots,d) \mid f_w(k,\dots, k) \\
   k &::= n  \mid r \mid g_w(k,\dots,k)\mid g_w(d,\dots,d)\mid k+k\mid -k\mid sum(t_K) \\
 t_D & ::=\parbox{4.8cm}{$D \mid f_y(t_D)\mid  f_y(t_K)$}
 t_K ::= Z \mid  Q \mid g_y(t_K) \mid g_y(t_D)
\end{align*}
\noindent where $v_b, v_s, v_n, v_r \in \varset$,
$\vartype(v_b)=\bool$, $b \inn \mathbb B$,
 $\vartype(v_n)=\integer$, $z \inn \mathbb Z$,
$\vartype(v_r)=\rational$, $q \inn \mathbb Q$, 
$\vartype(v_s)=\fset_i$, $h \inn {\mathbb D}_i$ (some $i$), $Z,Q,D \in\varset_{list}$, $\vartype(Z)=[\integer]$, $\vartype(Q)=[\rational]$, $D$ has non-arithmetic type, $f_{w}, f_{y}$ are functions with arithmetic codomains,  $g_w$, $g_y$ are functions with non-arithmetic ones.
\end{definition}

This definition may seem quite involved, but it captures essentially simple concepts. 
Term $s$ defines a string as a variable, constant, or inductive function application. For expressions of integer type, term $n$ allows variables, integers, or aggregators \emph{sum}, \emph{min}, and \emph{max} applied to lists of integers. Term $r$ is analogous to $n$ but for rationals, for which also the aggregator \emph{mean} is defined. Terms $k$ and $d$ define, by mutual induction, mixed terms that can combine different types: the only difference is that the root symbol for $k$ lives in a arithmetical domain ($\mathbb Z$ or $\mathbb Q$), wheres for $d$ it lives in a non-arithmetical domain. An analogous 
mutual
induction defines the list terms $t_D$ and $t_K$, which are built from list variables and functions, but differ in the fact that $t_K$ lives in a arithmetical domain.  Here a function applied to a list term is applied component-wise, and returns another list.
Notice also that a term $k$ can be produced by applying aggregator \emph{sum} to a list variable $t_k$.
Standard equivalences apply, hence disjunction (i.e., $\lor$) of constraints can be used, as well as comparisons $=$, $\neq$, $<$, $\leq$ on integer and rational expressions.
The set  of variables in a constraint $\varphi$ is denoted $\varsin{\varphi}$, and the set of all constraints over variables $\allvars$ by $\mathcal C(\allvars)$.

\begin{example}
\label{exa:conditions}
We consider two constraints that will express transition guards for our running example. 
First, let $d$ be an integer variable representing the maximum number of days expected for delivery by a customer, and $m$ a string variable denoting the shipment mode of an order. Constraint $(d \leq 5 \land m = \textsf{car}) \lor (d > 5 \land m = \textsf{truck})$ expresses that either $d$ is at most 5 days and the shipment mode is \textsf{car}, or that $d$ is 6 days or more and the shipment mode is \textsf{truck}.
Second, consider a list variable $P$ for products and a unary function $\mathit{cost}$ that returns the cost of each product, a rational number. This expresses the background knowledge that every product has a cost, that is however not explicitly manipulated by the process (so it will not appear in the log). Consistently with Def.~\ref{def:constraints}, $\mathit{cost}(P)$ represents the list 
that contains the costs of all elements in $P$,
and
$sum(cost(P)) \leq 1000$ expresses that the overall cost of all products in $P$ does not exceed 1000\texteuro.
\end{example}

\begin{definition}[DOPID]
\label{def:OCInet}
A  \emph{data-aware object-centric Petri net with identifiers} (\mynet) 
is a tuple  
$\NN = (\otypes, \Sigma_{val}, \places, \transitions, \inflow, \outflow, 
\coloring,\ell,\guass)$,
where:
\begin{compactenum}
\item $\places$ and $\transitions$ are finite sets of places and transitions such that $P\cap T=\emptyset$;
\item $\coloring\colon \places \rightarrow \colors$ maps every place to a color over $\Sigma$; 
\item $\ell \colon T \to \activities \cup \{\tau\}$ is the transition labelling where $\tau$ marks an invisible activity,
\item $\inflow \colon \places \times \transitions \rightarrow \Omega$ is a partial function called \emph{input flow} that satisfies $\coloring(\inflow(p,t))=\coloring(p)$
for every $(p,t)\in \dom(\inflow)$;
\item $\outflow\colon \transitions \times \places \rightarrow \Omega$ is a partial function called \emph{output flow} that satisfies $\coloring(\outflow(t,p))=\coloring(p)$ for every $(t,p)\in \dom(\outflow)$;
\item $\guass\colon  \transitions \rightarrow \mathcal C(\allvars)$ is a partial functions assigning guards, such that for every $t\in\transitions$ and $\guass(t)=\varphi$, $\varsin{\varphi}\subseteq\invars{t} \cup \outvars{t}$,
where $\invars{t}=\cup_{p\in\places}\varsin{\inflow(p,t)}$ and $\outvars{t}=\cup_{p\in\places}\varsin{\outflow(t,p)}$. 
\end{compactenum}
As a well-formedness condition, we assume that in $F_{in}$ one can only use only simple, $\subseteq$-template and  $=$-template inscriptions, while in $F_{out}$ one can only use simple and transfer-template inscriptions (cf.~Def.~\ref{def:inscription}).
\end{definition}

For a \mynet $\NN$ as in \defref{OCInet}, we also use the common notations for presets $\pre t = \{p \mid (p,t)\in \dom(\inflow)\}$ and postsets $\post t = \{p \mid (t,p)\in \dom(\outflow)\}$.

Simple flows (i.e., flows with simple inscriptions) are meant to consume and produce 
single tokens, whereas template flows (i.e., flows with template inscriptions) to consume and produce multiple matching tokens. Consumption in this case can be fine-tuned by indicating whether some (in the case of a $\subseteq$-template inscription) or all (for a $=$-template inscription) matching tokens have to be consumed. Production transfers such matched tokens to the corresponding output places, using transfer-template inscriptions.    
As illustrated by the next example and clarified by the definition of the semantics of \mynets, this is used to capture variable arcs in \cite{LissAA23}, but also to reconstruct different forms of synchronization. 
In particular, $=$-template inscriptions realize a form of data-aware wholeplace operations: they do not consume all tokens contained in a place, but all those that match the inscription.  
%
%
Ex.~\ref{exa:order} illustrates the most important features of \mynets;
it would not be expressible in
existing object-centric  formalisms. 

\begin{example}
\label{exa:order}
\figref{OPID:package:handling} graphically depicts a \mynet for the simple yet sophisticated order-to-ship process informally described in Ex.~\ref{exa:intro}.
%
Variables $\nu_o$ of type \emph{order} and $\nu_p$ of type \emph{product}, both in $\nuvarset$, refer to new orders and products. Normal variables $o,p \in \varset$ of type \emph{order} and \emph{product} refer to existing orders and products, and variable $P$ of type $\listtype{\mathit{product}}$ to lists of products. We also use the data value variables $d$ and $m$ described in Ex.~\ref{exa:conditions}.
For readability, single-component tuples  are written without brackets (e.g., we write $o$ instead of $\tup{o}$).

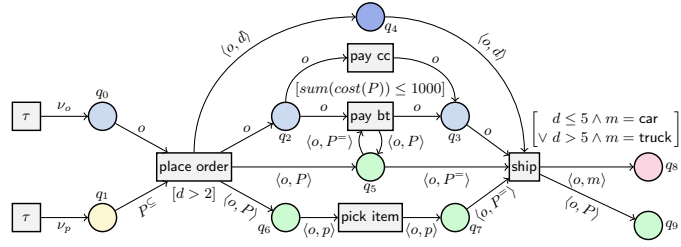
\begin{figure}[t]
\centering
\resizebox{.75\textwidth}{!}{
\begin{tikzpicture}[node distance=5mm and 8mm]

\node[oplace] (o0) {};
\node[below= of o0] (mid1) {};
\node[trans, left=of o0] (geno) {$\tau$};
 \draw[arc] (geno) -- node[above, insc] {$\nu_o$} (o0);

 \node[iplace, below=of mid1, label] (i0) {};
 \node[trans, left=of i0] (geni) {$\tau$};

 \draw[arc] (geni) -- node[below, insc] {$\nu_p$} (i0);

\node[trans,right=of mid1] (po) {$\m{place\ order}$};
\node[below=0mm of po,insc] {$[d>2]$};

 \draw[arc] (o0) -- node[above, insc] {$o$} (po);
 \draw[arc] (i0) -- node[below, insc, sloped] {$P^\subseteq$} (po);

\node[right=of po] (mid2) {};

\node[oplace,above=of mid2] (o1) {};
 \draw[arc] (po) -- node[above, insc] {$o$} (o1);

\node[trans,right=of o1] (ot1) {$\m{pay\ bt}$};
\node[above=0mm of ot1,insc] {$[sum(cost(P))\leq1000]$};

\node[trans,above=of ot1] (ot2) {$\m{pay\ cc}$};

 \draw[arc,out=90,in=180] (o1) edge node[above, insc] {$o$} (ot2);

 \draw[arc] (o1) -- node[above, insc] {$o$} (ot1);

\node[oiplace, at=(ot1|-po)] (corset) {};
 \draw[arc] (po) -- node[below, insc] {$\tup{o,P}$} (corset);

 \draw[arc,out=120,in=-120] (corset) edge node[left, insc] {$\tup{o,P^=}$} (ot1);
\draw[arc,out=-60,in=60] (ot1) edge node[right, insc] {$\tup{o,P}$} (corset);

 \node[oplace,right=of ot1] (o2) {};
 \draw[arc] (ot1) edge node[above, insc] {$o$} (o2);

 \draw[arc,out=0,in=90] (ot2) edge node[above, insc] {$o$} (o2);

  \node[oiplace,below=of mid2] (i1) {};
 \draw[arc] (po) -- node[below, insc, sloped] {$\tup{o,P}$} (i1);


 \node[trans,at=(i1-|corset)] (it1) {$\m{pick\ item}$};
 \draw[arc] (i1) -- node[below, insc] {$\tup{o,p}$} (it1);

 \node[oiplace,at=(it1-|o2)] (i2) {};
 \draw[arc] (it1) -- node[below, insc] {$\tup{o,p}$} (i2);

\node[odplace,above=2mm of ot2] (od) {};
 \draw[arc,out=90,in=180] (po) edge node[above, insc,sloped] {$\tup{o,d}$} (od);

\node[below=of o2] (mid3) {};
 
 \node[trans,right=of mid3] (ship) {$\m{ship}$};
\node[above=0mm of ship,insc,anchor=south west] {
$\left[\begin{array}{@{}rl} & d \leq 5 \land m = \textsf{car}\\ {} \lor & d > 5 \land m = \textsf{truck}\end{array}\right]$};

 \draw[arc,out=0,in=90] (od) edge node[above, insc,sloped] {$\tup{o,d}$} (ship);

\draw[arc] (o2) -- node[above,insc] {$o$} (ship);

\draw[arc] (corset) -- node[below,insc] {$\tup{o,P^=}$} (ship);

\draw[arc] (i2) -- node[below,insc,sloped] {$\tup{o,P^=}$} (ship);

\node[osplace,right=of ship,xshift=10mm] (os) {};
\draw[arc] (ship) -- node[below,insc] {$\tup{o,m}$} (os);

\node[oiplace,below=of os] (oi) {};
\draw[arc] (ship) -- node[below,insc,sloped] {$\tup{o,P}$} (oi);

\node[above=-.5mm of o0,insc] {$q_0$};
\node[above=-.5mm of i0,insc] {$q_1$};
\node[below=-.5mm of o1,insc] {$q_2$};
\node[below=-.5mm of o2,insc] {$q_3$};
\node[right=-1mm of od,yshift=-2mm,insc] {$q_4$};
\node[below=-.5mm of corset,insc] {$q_5$};
\node[left=-1mm of i1,yshift=-2mm,insc] {$q_6$};
\node[right=-1mm of i2,yshift=-2mm,insc] {$q_7$};
\node[right=-.5mm of os,insc] {$q_8$};
\node[right=-.5mm of oi,insc] {$q_9$};

\end{tikzpicture}
}
\caption{\mynet of an order-to-ship process.\label{fig:OPID:package:handling}}
\end{figure}

We explain the model transition by transition. The two silent transitions on the left have the purpose of injecting fresh orders and products in the net. 
Transition \textsf{place\:order} takes an order $o$ from place $q_0$ and \emph{some} available products $P$ from place $q_1$
that are assigned
to the order. 
Here the output transfer-template inscription $\tup{o,P}$,  transfers to place $q_5$ $|P|$ tokens, each carrying a pair $\tup{o,p}$ with $p$ taken from $P$. In this respect, place $q_5$ explicitly represents what in proclets is called \emph{correlation set}, describing for every  order in the system, which products belong to it. After firing \textsf{place\:order}, also $q_6$ contains products of $o$, but from there single products will be consumed  
independently through \textsf{pick\:item}. Besides its products, \textsf{place\:order} assigns to the order $o$ also the maximum number of days of delivery $d$, inserting the pair $\tup{o,d}$ in place $q_4$, responsible for tracking this attribute for every active order. Since $d$ is only used in output flows, this reconstructs what in DPNs is called a ``write'' variable, which is moreover constrained by condition $d > 2$, capturing that $d$ can only take values above 2 days. Finally, \textsf{place\:order} changes the state of the picked order $o$, moving it from place $q_0$ to $q_2$.

From $q_2$, two transitions can  be fired for $o$, reflecting two 
payment modes. Specifically, order $o$ either flows through \textsf{pay\:cc}, or through \textsf{pay\:bt}, but the latter can only be selected if the overall cost of $o$ 
does not exceed 1000{\texteuro} (cf.~\exaref{conditions}). To obtain all products of $o$, \textsf{pay\:bt} needs to fetch those products from place $q_5$, using inscription $\tup{o,P^=}$. This inscription requires that the first component matches order $o$ consumed from place $q_2$, while $P^=$ forces all matching pairs for $o$ to be included. As the aim is to use the products, but not to remove the corresponding pairs, they are all transferred back to $q_5$ using the 
inscription $\tup{o,P}$. 

Concurrently with order payment, the state of single products (recalling their order) is changed when they are, one by one, picked via the \textsf{pick\:item} transition. 

Finally, the \textsf{ship} transition is enabled for a paid order $o$ under the following conditions:
First, \emph{all} its products must have been picked. This is expressed through the two $=$-template input flows with the same inscription $\tup{o,P^=}$, which has the effect of consuming all pairs containing products of $o$ from places $q_5$ and $q_7$. The output transfer-template inscription $\tup{o,P}$ 
to place $q_9$ has the effect of transferring all those pairs there. At the same time, \textsf{ship} considers the value $d$ for $o$, and through the attached constraint (cf.~\exaref{conditions}) determines the shipment mode for $o$, which is recorded in place $q_8$, linking $m$ to $o$ using inscription $\tup{o,m}$. 
\end{example}

Two important remarks are in place wrt. \exaref{order}.
First, the modelling pattern in \figref{OPID:package:handling} that employs the  ``correlation'' place $q_5$ to keep track of products contained in an order, paired with the two $=$-template inscriptions in shipment to ensure that \emph{all} products of an order have been actually picked, is what makes \mynets able to support exact synchronization in the full generality of synchronous proclets \cite{Fahland19}, something that was out of reach until now for formal models based on PNIDs.
Second, as \mynets handle multiple objects and data values at once, they are not only able to express read-write conditions and guards as in DPNs \cite{LeoniFM18}, but also more sophisticated conditions using aggregation expressions.

\smallskip
\noindent
\textbf{Semantics.}
Given the set of object ids $\objects$ and a set of data-values $\mathcal {DV}$,
the set of \emph{tokens} $\objtuples$ is the set of tuples that consist of object ids and data values $\objtuples\,{=}\,\{(\objects\uplus\mathcal {DV})^m\,{\mid}\,m\,{\geq}\,1\}$.
The \emph{color} of a token $\omega\in \objtuples$ of the form 
$\omega=\tup{d_1, \dots, d_m}$ 
is given by
$\coloring(\omega)=\tup{\vartype(d_1), \dots, \vartype(d_m)}$.
To define the execution semantics, we first introduce a notion of
a \emph{marking} of a \mynet $\NN=\tup{\otypes, \Sigma_{val},\places, \transitions, \inflow, \outflow, \coloring,\ell,\guass}$, namely as a function $\marking\colon\places\rightarrow 2^{\objtuples}$, such that for all $p\in P$ and $\tup{d_1, \dots, d_m} \in \marking(p)$, it holds that $\coloring(\tup{d_1, \dots, d_m}) = \coloring(p)$. 
Let $\mathit{Lists}(\objects)$ denote the set of object lists of the form $[o_1, \dots, o_k]$ with $o_1,\dots, o_k \in \objects$ such that all $o_i$ have the same object type;
the type of such a list is then $\listtype{\vartype(o_1)}$. Analogously, given $\mathit{Lists}(\mathcal{DV})$ the set of data-value lists  $[dv_1, \dots, dv_l]$ with $dv_1,\dots, dv_l \in \mathcal{DV}$ such that all $dv_i$ have the same data-value type,
the type of such a list is  $\listtype{\vartype(dv_1)}$.
%
Next, we define \emph{bindings} to fix which data are involved in a transition firing.

\begin{definition}[Binding]
A \emph{binding} for a transition $t$ and a marking $\marking$ is a type-preserving function 
$b\colon \invars{t} \cup \outvars{t} \to (\objects \cup \mathit{Lists}(\objects))\uplus (\mathcal{DV}\cup \mathit{Lists}(\mathcal{DV}))$,
such that for all $U \in \listvarset$, we have $b(U)=b(U^=)=b(U^\subseteq)$.
To ensure freshness of created values, we demand that $b$ is injective on $\nuvarset \cap \outvars{t}$, and that $b(\nu)$ does not occur in $\marking$ for all $\nu \in \nuvarset \cap \outvars{t}$.
\end{definition}

E.g., for transition \textsf{ship} in \exaref{order} the mapping $b$ that sets
$b(o) = \co{o}_1$ and $b(P)=[\co{p}_1,\co{p}_2,\co{p}_3]$ is a binding.
Next, we extend bindings to inscriptions to fix which tokens 
participate in a transition firing.
The extension of a binding $b$ to inscriptions, i.e., variable tuples, is denoted $\vec b$.
For an inscription $\iota\,{=}\,\tup{v_1, \dots, v_m}$
and binding $b$ such that $o_i\,{=}\,b(v_i)$ for all $1\,{\leq}\,i\,{\leq}\,m$,
let $\vec b(\iota)$ be the set of object tuples defined as follows:
if $\iota$ is a simple inscription then
$\vec b(\iota) = \{\tup{o_1, \dots, o_m}\}$.
Otherwise, there must be one $v_i$, $1\,{\leq}\,i\,{\leq}\,n$, such that $v_i \in \listvarset$, and consequently $o_i$ must be a list, say $o_i=[u_1, \dots, u_k]$ for some $u_1, \dots, u_k$.
Then $\vec b(\iota) = \{\tup{o_1, \dots, o_{i-1}, u_1, o_{i+1},\dots, o_m}, \dots, \tup{o_1, \dots, o_{i-1}, u_k, o_{i+1},\dots, o_m}\}$.
The set of all bindings is denoted by $\mathcal B$.
We next define when a transition together with a binding is enabled in a marking.

\begin{definition}[Enablement]
\label{def:enabled}
A transition $t \in\transitions$ and a binding $b$ for marking $\marking$ are \emph{enabled} in $\marking$ if $b(guard(t))$ is satisfiable, for all $p \in \pre{t}$,
$\vec b(\inflow(p,t)) \subseteq \marking(p)$ and
if $\inflow(p,t)$ is a $=$-variable flow with list variable $V^=$ there is no binding $b'$ that differs from $b$ only wrt. $V^=$ s.t.
$\vec b'(\inflow(p,t)) \subseteq \marking(p)$ and
$b(V^=) \subset b'(V^=)$.
\end{definition}
 
E. g., the binding $b$ with $b(o) = \co{o}_1$ and $b(P)=[\co{p}_1,\co{p}_2,\co{p}_3]$ is enabled in marking $\marking$ of the net in \exaref{order} with $\tup{\co{o}_1}\in \marking(q_{\mathit{blue}})$ and $\tup{\co{o}_1,\co{p}_1}, \tup{\co{o}_1,\co{p}_2}, \tup{\co{o}_1,\co{p}_3}\in \marking(q_{\mathit{green}})$, for $q_{\mathit{blue}}$ and $q_{\mathit{green}}$ the input places of \textsf{ship} with respective color.
 
\begin{definition}[Firing]
Let transition $t$ and binding $b$ be enabled in marking $\marking$.
The \emph{firing} of $t$ with $b$ yields the marking $\marking'$ given by $\marking'(p)=\marking(p) \setminus \vec b(\inflow(p,t)) $ for all $p \in {\pre t} \setminus {\post t}$,
$\marking'(p)=\marking(p) \cup \vec b(\outflow(p,t))$ for all $p \in {\post t}\setminus {\pre t}$, and $\marking'(p)=\marking(p)$ for all $p \in {\post t}\cap {\pre t}$.
\end{definition}

We write $\marking \goto{t,b} \marking'$ to denote that $t$ is enabled with binding $b$ in $\marking$, and its firing yields $\marking'$.
A sequence of transitions with bindings 
$\run = \tup{(t_1, b_1), \dots, (t_n, b_n)}$ is called a \emph{run} 
if $\marking_{i-1} \goto{t_i, b_i} \marking_i$ for all $1\leq i \leq n$,
in which case we write $\marking_0 \goto{\run} \marking_n$.
For such a binding sequence $\run$, the \emph{visible subsequence} $\run_v$ is the subsequence of $\run$ consisting of all $(t_i, b_i)$ such that $\ell(t_i) \neq \tau$.

An \emph{accepting} object-centric Petri net with identifiers is an object-centric Petri net $\NN$
together with a set of initial markings $M_{\mathit{init}}$ and a set of final markings  $M_{\mathit{final}}$.
For instance, for \exaref{order}, $M_{\mathit{init}}$ consists only of the empty marking, whereas $M_{\mathit{final}}$ consists of all (infinitely many) markings in which each of the two right-most places has at least one token, and all other places have no token.
The \emph{language} of the net is given by
$\LL(\NN) = \{\run_v \mid m \goto{\run} m',\ m \in M_{\mathit{init}}\text{, and } m' \in M_{\mathit{final}}\}$,
i.e., the set of visible subsequences of accepted sequences.

The next example relates an observed event log with a \mynet, preluding to the conformance checking problem tackled in the next section.
\begin{example}
The event log described in \exaref{log} cannot be suitably replayed in the \mynet $\NN$ of \figref{OPID:package:handling}, due to two mismatches: according to $\NN$, $\co{o}_1$ must be shipped by \textsf{car} (as the preferred days are below 5), and with both products $\co{p}_1$ and $\co{p}_2$. This in turn requires that, before shipping, also product $\co{p}_1$ must be picked.
\end{example}

\smallskip
\noindent
\textbf{Modelling considerations.}
We briefly relate \mynets to the two reference formalisms that infuse Petri nets with case attributes (namely DPNs \cite{LeoniFM18}) and multiple objects with complex synchronization mechanisms (namely sychronous proclets \cite{Fahland19}). DPNs can be expressed as \mynets using an approach similar to the encoding of DPNs into Colored Petri nets in \cite{LeoniFM18}, 
using a ``data place'' for each variable $x$ that contains a single token carrying the current value of $x$, and is linked to all transitions that read or write $x$. 

Proclets are structurally encoded into \mynets following a  schema similar to the one described for OPIDs \cite{GianolaMW24}. Since \mynets provide full support for subset and exact synchronization, the crux is to refine the approach in \cite{GianolaMW24} to 
reflect
correlation sets, and their consequent usage for synchronization. This is done as follows: for every correlation set linking multiple child objects of the many side to the single parent object of the one side, a special ``correlation'' place holding the  pairs is introduced. Upon synchronization, this correlation place is inspected to extract some or all the required pairs. This  reconstructs and generalizes proclet synchronization, as one can now operate over the correlation place to define different regeneration strategies for the correlation set.

As for more general modelling languages, we leave as future work to provide a systematic formalization into \mynets. This appears to be a  feasible route, building on previous encodings of artifact-centric and case-handling approaches into (extensions of) Petri nets \cite{Lohmann,Weske}.

\section{Alignment-Based Conformance Checking for \mynets} 

We follow alignment-based approaches for object-centric processes \cite{LissAA23,GianolaMW24}, 
which  relate trace graphs to model runs to find deviations.
In the sequel, we consider a trace graph $\tracenet_X$ and an accepting \mynet $\NN$, 
assuming that the language of $\NN$ is not empty.
In our data-aware setting, moves also contain assignments:

\begin{definition}[Moves]
\label{def:moves}
A  \emph{model} move is a tuple in
$\{\SKIP\} \times ((\activities \cup \{\tau\}) \times\mathcal P(\objects) \times \assignments)$,
a \emph{log move} a tuple in
$(\activities \times\mathcal P(\objects) \times \assignments) \times \{\SKIP\}$, and 
a \emph{synchronous} move is of the form
$\tup{\tup{a,O_M,\alpha_M},\tup{a',O_L,\alpha_L}} \in (\activities \times\mathcal P(\objects) \times \assignments)^2 $ such that $a=a'$ and $O_L = O_M$.
The set of all synchronous, model, and log moves over $T_X$ and $\NN$ is denoted $\moves(\tracenet_X, \NN)$.
\end{definition}

In the object-centric setting, an alignment is a \emph{graph} of moves $G$. 
We use the notions of log projection $G\restrlog$ and model projection $G\restrmod$  as defined in~\cite{LissAA23,GianolaMW24}, 
but provide an intuitive explanation here:
%
%
%
%
The log projection is the graph obtained from $G$ by projecting nodes to their log part, while omitting model moves, i.e. nodes where the log component is $\SKIP$.
Edges are as in $G$, except that edges are added that ``shortcut'' over model moves in $G$. The model projection is defined similarly.
%
Next we define an alignment as a graph over moves where the log and model projections are a trace graph and a run, respectively.

\begin{definition}[Alignment]
An \emph{alignment} of a trace graph $\tracenet_X$ and an accepting \mynet $\NN$ is an acyclic directed graph $\Gamma=\tup{C,B}$ with $C \subseteq \moves(\tracenet_X, \NN)$
such that $\Gamma\restrlog=\tracenet_X$, 
there is a run $\run = \tup{\tup{t_1, b_1}, \dots, \tup{t_n,b_n}}$ with $\run_v \in \LL(\NN)$, and the model projection $\Gamma\restrmod = \tup{C_m,B_m}$ admits a bijection $f\colon \{\tup{t_1, b_1}, \dots, \tup{t_n,b_n}\} \to C_m$ such that
\begin{compactitem}[$\bullet$]
\item 
if $f(t_i, b_i) = \tup{a, O_M, \alpha_M}$ then $\ell(t_i)=a$, $O_M= \mathit{range}(b_i) \cap \objects$, and 
$\alpha_M = \{ x \mapsto d \mid x\in \Dom(b_i)\text{, }b_i(x)=d\text{, and }\vartype(x) \in \Sigma_{val} \}$;
\item
for all $\tup{r,r'}\in B_m$ there are $1{\leq}i{<}j{\leq}n$ such that $f(t_i,b_i){=}r$ and $f(t_j,b_j){=}r'$.
\end{compactitem}
\end{definition}

\begin{example}
\label{exa:alignment}
Below is an alignment $\Gamma$ for the log in \exaref{log} wrt.~$\NN$ in \exaref{order}. The log (resp. model) component is shown on top (resp. bottom) of moves.\\
\resizebox{\textwidth}{!}
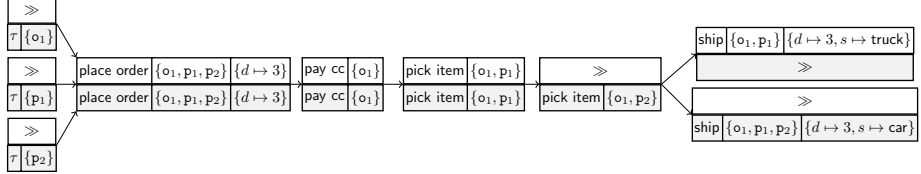
{
\begin{tikzpicture}[xscale=1.4, yscale=.7]
\begin{scope}
\node[trans, model, anchor=north,minimum width=11mm,splitme] (co1) at (-2,1.6) {$\tau$\nodepart{two}$\{\co{o}_1\}$};
\node[trans, log, anchor=south,minimum width=11mm] (co1s) at (-2,1.6) {$\SKIP$};
\node[trans, model, anchor=north,minimum width=11mm,splitme] (ci1) at (-2,0) {$\tau$\nodepart{two}$\{\co{p}_1\}$};
\node[trans, log, anchor=south,minimum width=11mm] (ci1s) at (-2,0) {$\SKIP$};
\node[trans, model, anchor=north,minimum width=11mm,splitme] (ci2) at (-2,-1.6) {$\tau$\nodepart{two}$\{\co{p}_2\}$};
\node[trans, log, anchor=south,minimum width=11mm] (ci2s) at (-2,-1.6) {$\SKIP$};
\node[trans, model, anchor=north,splitmemore] (pom) {$\m{place\ order}$\nodepart{two}$\{\co{o}_1,\co{p}_1,\co{p}_2\}$ \nodepart{three}$\{d \mapsto 3\}$};
\node[trans, log, anchor=south,splitmemore] (pot) {$\m{place\ order}$\nodepart{two}$\{\co{o}_1,\co{p}_1,\co{p}_2\}$ \nodepart{three}$\{d \mapsto 3\}$};
\node[trans, model, anchor=north,splitme] (pm) at (2.1,0) {$\m{pay\ cc}$\nodepart{two}$\{\co{o}_1\}$};
\node[trans, log, anchor=south,splitme] (pt) at (2.1,0) {$\m{pay\ cc}$\nodepart{two}$\{\co{o}_1\}$};
\node[trans, model, anchor=north,splitme] (pi1m) at (3.7,0) {$\m{pick\ item}$\nodepart{two}\parbox{12mm}{$\{\co{o}_1,\co{p}_1\}$}};
\node[trans, log, anchor=south,splitme] (pi1t) at (3.7,0) {$\m{pick\ item}$\nodepart{two}\parbox{12mm}{$\{\co{o}_1,\co{p}_1\}$}};
\node[trans, model, anchor=north,splitme] (pi2m) at (5.5,0) {$\m{pick\ item}$\nodepart{two}\parbox{12mm}{$\{\co{o}_1,\co{p}_2\}$}};
\node[trans, log, anchor=south,minimum width=28mm] (pi2t) at (5.5,0) {$\SKIP$};
\node[trans, model, anchor=north,minimum width=50mm] (sm) at (8.2,.8) {$\SKIP$};
\node[trans, log, anchor=south,minimum width=22mm,splitmemore] (st) at (8.2,.8) {$\m{ship}$\nodepart{two}$\{\co{o}_1, \co{p}_1\}$\nodepart{three}$\{d \mapsto 3, s\mapsto\m{truck}\}$};
\node[trans, model, anchor=north,minimum width=26mm,splitmemore] (sma) at (8.2,-.8) {$\m{ship}$\nodepart{two}$\{\co{o}_1, \co{p}_1,\co{p}_2\}$\nodepart{three}$\{d \mapsto 3, s\mapsto\m{car}\}$};
\node[trans, log, anchor=south,minimum width=52mm] (sta) at (8.2,-.8) {$\SKIP$};
\draw[arc] (pom.north east) -- (pm.north west);
\draw[arc] (pm.north east) -- (pi1m.north west);
\draw[arc] (pi1m.north east) -- (pi2m.north west);
\draw[arc] (co1.north east) -- (pot.north west);
\draw[arc] (ci1.north east) -- (pom.north west);
\draw[arc] (ci2.north east) -- (pom.south west);
\draw[arc] (pi2m.north east) -- (sm.north west);
\draw[arc] (pi2m.north east) -- (sma.north west);
\end{scope}
\end{tikzpicture}}
Note that a synchronous $\mathsf{ship}$  move is not possible by Def.~\ref{def:moves} because the sets of involved objects would differ.
\end{example}

We adopt the cost function from \cite{LissAA23}, but extend it to account for mismatching data values.
Other definitions are, however, possible as well.

\begin{definition}[Cost]
\label{def:cost}
The cost of a move is: 
\begin{inparaenum}[\it (1)]
\item
if $M=\tup{\tup{a_L, O_L,\alpha_L}, \SKIP}$ is a log move then $\cost(M) = |O_L|+|\dom(\alpha_L)|$, 
\item
if $M=\tup{\SKIP, \tup{a_M, O_M,\alpha_M}}$ is a model move  then $\cost(M) = 0$ if $a_\mod=\tau$, and $\cost(M) = |O_M|+|\dom(\alpha_M)|$ 
otherwise, 
\item
if $M$ is a synchronous move $\tup{\tup{a_L, O_L, \alpha_L}, \tup{a_M, O_M, \alpha_M}}$ then $\cost(M)$ is the number of variables in $\dom(\alpha_L)\cup \dom(\alpha_M)$ for which 
$\alpha_L$ and $\alpha_M$ differ.
\end{inparaenum}
For an alignment $\Gamma=\tup{C,B}$, we set $\cost(\Gamma) =\sum_{M\in C}\cost(M)$, i.e., the cost of an alignment $\Gamma$ is  the sum of the cost of its moves.
\end{definition}

E.g., $\Gamma$ in \exaref{alignment} has cost 11, as it involves one log move (cost 4) and two non-silent model moves (costs 2 and 5).
In fact, $\Gamma$ is optimal:

\begin{definition}[Optimal alignment]
An alignment $\Gamma$ of a trace graph $\tracenet_X$ 
and an accept\-ing \mynet $\NN$ is \emph{optimal} 
if $\cost(\Gamma)\,{\leq}\,\cost(\Gamma')$ for all alignments $\Gamma'$ of $\tracenet_X$ and $\NN$.
\end{definition}

The \emph{conformance checking task} for an accepting \mynet $\NN$ and a log $L$ is
to find optimal alignments with respect to $\NN$ for all trace graphs in $L$.


\smallskip
\noindent
\textbf{SMT encoding for conformance checking.} 
An SMT encoding of the conformance checking task for a given \mynet $\NN$ and trace graph $T_X$ can be done in a similar way as for OPIDs \cite{GianolaMW24}. For reasons of space, we focus on the differences.

First, in encoding-based conformance checking, it is essential to fix upfront an upper bound on the size of an optimal alignment $\Gamma$. For \mynets, we can exploit \cite[Lemma 1]{GianolaMW24}: \mynets differ from OPIDs in the presence of data and synchronization, but this does not affect the reasoning of that proof.
We thus get an upper bound $n$ on the number of nodes in the model projection of $\Gamma$ and an upper bound $K$ on the number of objects used in a transition. From $T_X$ and $K$, we can get a finite set of objects $O$ such that $\Gamma$ uses only objects in $O$ (up to renaming). 
Let $m$ be the number of nodes in $T_X$.

\newcommand{\transvar}{\mathtt T}
\newcommand{\markvar}{\mathtt M}
\newcommand{\objvar}{\mathtt O}
\newcommand{\datavar}{\mathtt D}
\newcommand{\dstorevar}{\mathtt S}
\newcommand{\rlvar}{\mathtt {len}}
\newcommand{\distvar}{\delta}
The encoding uses the SMT variables from~\cite{GianolaMW24}:
\begin{inparaenum}[(a)]
\item transition variables $\transvar_i$, $1\leq i\leq n$,
to encode 
the $i$-th transition in the run;
\item marking variables $\markvar_{i,p,\vec o}$ for every time point $0\leq i\leq n$, every place $p$, and every vector $\vec o$ of objects with elements in $O$;
\item a variable $\rlvar$ to encode the length of the run and
\item object variables $\objvar_{i,k}$ for all $1\leq i\leq n$ and $0\leq k \leq K$ to encode which objects populate inscriptions, and
\item distance variables $\distvar_{i,j}$ to optimize the cost of the alignment. 
\end{inparaenum}
In addition, 
to keep track of data values, if $X$ is the set of inscription variables of non-object type in $\NN$, and $M$ the maximal number of data values in tokens, we use
\begin{inparaenum}[(a)]
\setcounter{enumi}{5}
\item a data inscription variable $\datavar_{i,x}$ to represent the data value of $x$ in the $i$-th transition, for all $1\leq i\leq n$ and $x\in X$; and
\item a data store variable $\dstorevar_{i,p,\vec o,l}$ to represent the $l$-th data value stored with token $\vec o$ in place $p$ at instant $i$, for all $1\leq l \leq M$, object vectors $\vec o$ over $O$, places $p$, and $1\leq i\leq n$.
\end{inparaenum}

There are then two main differences in the encoding wrt.~\cite{GianolaMW24}. First, transitions guards need to be taken into account, similar to~\cite{FelliGMRW23}, using the data variables $\datavar_{j,x}$.
Uninterpreted function symbols as well as 
numeric predicates and aggregation functions are natively supported by SMT solvers.
Second, to model synchronization, in contrast to the subset synchronization employed in~\cite{GianolaMW24} it must be ensured that inscription variables from $\listvarset^=$ are always instantiated by all matching tokens currently in the respective places.
Details of the encoding can be found in \cite{long-version}.
Notably, we show that from a satisfying assignment to all constraints, an optimal alignment for $\NN$ and $T_X$ can be decoded.

\smallskip
\noindent
\textbf{Implementation.}
We extended the conformance checker \textsf{CoCoMoT} (\texttt{\url{https://github.com/bytekid/cocomot}}) to support \mynets, using the SMT solver \textsf{Yices 2} as backend and the aforementioned encoding.
We tested it on a series of examples that can be found in the repository.
For Ex.~\ref{exa:order} and traces of length in the same scale as in the running example, conformance checking is done below one second.



\section{Conclusions}
We have introduced \mynets, a new process formalism that unifies modelling features of case-centric data-aware processes and object-centric processes, especially offering an object-centric paradigm with full synchronization and support for complex data. We also showed a novel operational approach leveraging the SMT technology to tackle alignment-based conformance checking for \mynets.
In future work, we intend to conduct an experimental evaluation of this approach, and study discovery techniques for \mynets.

%
%
\newpage
\begin{credits}
\subsubsection{Acknowledgements.}
M.~Montali was partially supported by the NextGenerationEU FAIR PE0000013 project MAIPM
(CUP C63C22000770006) and the PRIN MIUR project PINPOINT Prot. 2020FNEB27. S.~Winkler was partially supported by the UNIBZ project TEKE. A. Gianola was partly supported by Portuguese national funds through Fundação para a Ciência e a Tecnologia, I.P. (FCT), under projects UIDB/50021/2020 (DOI: 10.54499/UIDB/50021/2020). This work was partially supported by the `OptiGov' project, with ref. n. 2024.07385.IACDC (DOI: 10.54499/2024.07385.IACDC), fully funded by the `Plano de Recuperação e Resiliência' (PRR) under the investment `RE-C05-i08 - Ciência Mais Digital', measure `RE-C05-i08.m04' (in accordance with the FCT Notice No. 04/C05-i08/2024), framed within the financing agreement signed between the `Estrutura de Missão Recuperar Portugal' (EMRP) and FCT as an intermediary beneficiary.
\end{credits}
%
%
 \bibliographystyle{splncs04}
 \bibliography{references}

\begin{thebibliography}{10}
\providecommand{\url}[1]{\texttt{#1}}
\providecommand{\urlprefix}{URL }
\providecommand{\doi}[1]{https://doi.org/#1}

\bibitem{Aalst19}
van~der Aalst, W.M.P.: Object-centric process mining: Dealing with divergence
  and convergence in event data. In: Proc. {17th SEFM} (2019).
  \doi{10.1007/978-3-030-30446-1\_1}

\bibitem{Aals23}
van~der Aalst, W.M.P.: Toward more realistic simulation models using
  object-centric process mining. In: Proc. 37th {ECMS}. pp. 5--13 (2023).
  \doi{10.7148/2023-0005}, \url{https://doi.org/10.7148/2023-0005}

\bibitem{Aalst23}
van~der Aalst, W.M.P.: Twin transitions powered by event data - using
  object-centric process mining to make processes digital and sustainable. In:
  Joint Workshop Proc. {ATAED}/{PN4TT} (2023)

\bibitem{AalstB20}
van~der Aalst, W.M.P., Berti, A.: Discovering object-centric {P}etri nets.
  Fundam. Informaticae  \textbf{175}(1-4),  1--40 (2020).
  \doi{10.3233/FI-2020-1946}

\bibitem{AaWG05}
van~der Aalst, W.M.P., Weske, M., Gr{\"{u}}nbauer, D.: Case handling: a new
  paradigm for business process support. Data Knowl. Eng.  \textbf{53}(2),
  129--162 (2005). \doi{10.1016/J.DATAK.2004.07.003}

\bibitem{BeMA23}
Berti, A., Montali, M., van~der Aalst, W.M.P.: Advancements and challenges in
  object-centric process mining: {A} systematic literature review. CoRR
  \textbf{abs/2311.08795} (2023). \doi{10.48550/ARXIV.2311.08795}

\bibitem{BoltenhagenCC21}
Boltenhagen, M., Chatain, T., Carmona, J.: Optimized {SAT} encoding of
  conformance checking artefacts. Computing  \textbf{103}(1),  29--50 (2021)

\bibitem{Breitmayer0PR24}
Breitmayer, M., Arnold, L., Pejic, M., Reichert, M.: Transforming
  object-centric process models into {BPMN} 2.0 models in the
  {PHILharmonicFlows} framework. In: Proc. Modellierung 2024. {LNI}, vol.
  {P-348}, pp. 83--98 (2024). \doi{10.18420/MODELLIERUNG2024\_009}

\bibitem{CalvaneseGGMR19}
Calvanese, D., Ghilardi, S., Gianola, A., Montali, M., Rivkin, A.: Formal
  modeling and {SMT}-based parameterized verification of data-aware {BPMN}. In:
  Proc. of {BPM} 2019. LNCS, vol. 11675, pp. 157--175 (2019),
  \url{https://doi.org/10.1007/978-3-030-26619-6\_12}

\bibitem{CohH09}
Cohn, D., Hull, R.: Business artifacts: {A} data-centric approach to modeling
  business operations and processes. {IEEE} Data Eng. Bull.  \textbf{32}(3),
  ~3--9 (2009)

\bibitem{DamaggioDHV11}
Damaggio, E., Deutsch, A., Hull, R., Vianu, V.: Automatic verification of
  data-centric business processes. In: Proc. of {BPM} 2011. LNCS, vol.~6896,
  pp. 3--16 (2011). \doi{10.1007/978-3-642-23059-2\_3}

\bibitem{Fahland19}
Fahland, D.: Describing behavior of processes with many-to-many interactions.
  In: Proc. {PETRI NETS} (2019). \doi{10.1007/978-3-030-21571-2\_1}

\bibitem{FelliGMRW23}
Felli, P., Gianola, A., Montali, M., Rivkin, A., Winkler, S.: Data-aware
  conformance checking with {SMT}. Inf. Syst.  \textbf{117},  102230 (2023).
  \doi{10.1016/J.IS.2023.102230}

\bibitem{FelliGMRW23b}
Felli, P., Gianola, A., Montali, M., Rivkin, A., Winkler, S.: Multi-perspective
  conformance checking of uncertain process traces: An {SMT}-based approach.
  Eng. Appl. Artif. Intell.  \textbf{126},  106895 (2023).
  \doi{10.1016/J.ENGAPPAI.2023.106895},
  \url{https://doi.org/10.1016/j.engappai.2023.106895}

\bibitem{FelliLM21}
Felli, P., de~Leoni, M., Montali, M.: Soundness verification of data-aware
  process models with variable-to-variable conditions. Fundam. Informaticae
  \textbf{182}(1),  1--29 (2021). \doi{10.3233/FI-2021-2064}

\bibitem{GhilardiGMR21}
Ghilardi, S., Gianola, A., Montali, M., Rivkin, A.: {Delta-BPMN}: {A} concrete
  language and verifier for data-aware {BPMN}. In: Proc. of {BPM} 2021. Lecture
  Notes in Computer Science, vol. 12875, pp. 179--196. Springer (2021).
  \doi{10.1007/978-3-030-85469-0\_13},
  \url{https://doi.org/10.1007/978-3-030-85469-0\_13}

\bibitem{GhilardiGMR22}
Ghilardi, S., Gianola, A., Montali, M., Rivkin, A.: Petri net-based
  object-centric processes with read-only data. Inf. Syst.  \textbf{107},
  102011 (2022). \doi{10.1016/J.IS.2022.102011}

\bibitem{Gianola23}
Gianola, A.: Verification of Data-Aware Processes via Satisfiability Modulo
  Theories, Lecture Notes in Business Information Processing, vol.~470.
  Springer (2023). \doi{10.1007/978-3-031-42746-6}

\bibitem{GianolaMW24}
Gianola, A., Montali, M., Winkler, S.: Object-centric conformance alignments
  with synchronization. In: Proc. 36th CAiSE. LNCS, vol. 14663, pp. 3--19
  (2024). \doi{10.1007/978-3-031-61057-8\_1}

\bibitem{long-version}
Gianola, A., Montali, M., Winkler, S.: Object-centric processes with structured
  data and universal synchronization (extended version) (2024), available from
  \url{ https://www.inf.unibz.it/montali/papers/dopid-long-version.pdf}

\bibitem{Weske}
Haarmann, S., Montali, M., Weske, M.: Refining case models using cardinality
  constraints. In: La~Rosa, M., Sadiq, S., Teniente, E. (eds.) Proc. 33rd
  CAiSE. pp. 296--310 (2021). \doi{10.1007/978-3-030-79382-1\_18}

\bibitem{HewW16}
Hewelt, M., Weske, M.: A hybrid approach for flexible case modeling and
  execution. In: Proc. Business Process Management Forum. LNBIP, vol.~260, pp.
  38--54 (2016). \doi{10.1007/978-3-319-45468-9\_3}

\bibitem{KunzleR11}
K{\"{u}}nzle, V., Reichert, M.: {PHILharmonicFlows:} towards a framework for
  object-aware process management. J. Softw. Maintenance Res. Pract.
  \textbf{23}(4),  205--244 (2011). \doi{10.1002/SMR.524},
  \url{https://doi.org/10.1002/smr.524}

\bibitem{LeoniFM18}
de~Leoni, M., Felli, P., Montali, M.: A holistic approach for soundness
  verification of decision-aware process models. In: {ER}. LNCS, vol. 11157,
  pp. 219--235 (2018). \doi{10.1007/978-3-030-00847-5\_17}

\bibitem{LissAA23}
Liss, L., Adams, J.N., van~der Aalst, W.M.P.: Object-centric alignments. In:
  Proc. ER (2023). \doi{10.1007/978-3-031-47262-6\_11}

\bibitem{Lohmann}
Lohmann, N., Wolf, K.: Artifact-centric choreographies. In: Service-Oriented
  Computing. pp. 32--46 (2010). \doi{10.1007/978-3-642-17358-5\_3}

\bibitem{MannhardtLRA16}
Mannhardt, F., de~Leoni, M., Reijers, H.A., van~der Aalst, W.M.P.: Balanced
  multi-perspective checking of process conformance. Computing  \textbf{98}(4),
   407--437 (2016). \doi{10.1007/S00607-015-0441-1}

\bibitem{MontaliC16}
Montali, M., Calvanese, D.: Soundness of data-aware, case-centric processes.
  Int. J. Softw. Tools Technol. Transf.  \textbf{18}(5),  535--558 (2016).
  \doi{10.1007/S10009-016-0417-2}

\bibitem{MonR17}
Montali, M., Rivkin, A.: {DB-Nets:} on the marriage of colored petri nets and
  relational databases. Trans. Petri Nets Other Model. Concurr.  \textbf{12},
  91--118 (2017). \doi{10.1007/978-3-662-55862-1\_5}

\bibitem{PWOB19}
Polyvyanyy, A., van~der Werf, J.M.E.M., Overbeek, S., Brouwers, R.: Information
  systems modeling: Language, verification, and tool support. In: Proc. 31st
  {CAiSE} (2019). \doi{10.1007/978-3-030-21290-2\_13}

\bibitem{RVFE10}
Rosa{-}Velardo, F., de~Frutos{-}Escrig, D.: Decidability problems in {Petri}
  nets with names and replication. Fundam. Informaticae  \textbf{105}(3),
  291--317 (2010). \doi{10.3233/FI-2010-368}

\bibitem{SVDD23}
Snoeck, M., Verbruggen, C., Smedt, J.D., Weerdt, J.D.: Supporting data-aware
  processes with {MERODE}. Softw. Syst. Model.  \textbf{22}(6),  1779--1802
  (2023). \doi{10.1007/S10270-023-01095-4}

\bibitem{SoSD22}
Sommers, D., Sidorova, N., van Dongen, B.: Aligning event logs to
  resource-con\-strained {\(\nu\)}-petri nets. In: Proc. 43rd\,{PETRI}\,{NETS}.
  LNCS, vol. 13288, pp. 325--345 (2022). \doi{10.1007/978-3-031-06653-5\_17}

\bibitem{HBGV13}
Terry Heath~III, F.F., Boaz, D., Gupta, M., Vacul{\'{\i}}n, R., Sun, Y., Hull,
  R., Limonad, L.: Barcelona: {A} design and runtime environment for
  declarative artifact-centric {BPM}. In: Proc. 11th ICSOC. LNCS, vol.~8274,
  pp. 705--709 (2013). \doi{10.1007/978-3-642-45005-1\_65}

\bibitem{WerfRPM22}
van~der Werf, J.M.E.M., Rivkin, A., Polyvyanyy, A., Montali, M.: Data and
  process resonance - identifier soundness for models of information systems.
  In: Proc. {PETRI NETS} (2022). \doi{10.1007/978-3-031-06653-5\_19}

\end{thebibliography}
\appendix
\extended{

\section{Encoding}

\newcommand\prodtoken{\mathit{produced}}
\newcommand\constoken{\mathit{consumed}}
\newcommand\synced{\mathit{synced}}

We detail the encoding outlined in the main body of the paper.
The encoding crucially relies on the following bound on the number of objects and the number of moves in an optimal alignment, taken from~\cite{GianolaMW24}. \mynets differ from the OPIDs in \cite{GianolaMW24} by the presence of data and synchronization, but this does not affect the reasoning of this proof.

\begin{lemma}
\label{lem:bounds:precise}
Let  $\NN$ be \anet and $T_X=\tup{E_X,D_X}$ a trace graph with optimal alignment $\Gamma$.
Let $m=\sum_{e\in E_X} |\projobj(e)|$ the number of object occurrences in $E_X$, and 
$c = \sum_{i=1}^n |\dom(b_i)|$ the number of object occurrences in some run  $\rho$ of $\NN$, with $\rho_v=\tup{\tup{t_1,b_1}, \dots \tup{t_n,b_n}}$.
Then $\Gamma|_\mod$ has 
at most $(|E_X|+c+m)(k+1)$ moves if $\NN$ has no $\nu$-inscriptions, and at most $(|E_X|+3c+2m)(k+1)$ otherwise, where $k$ is the longest sequence of silent transitions without $\nu$-inscriptions in $\NN$.
Moreover, $\Gamma|_\mod$ has at most $2c+m$ object occurrences in non-silent transitions.
\end{lemma}

Next, we detail which variables are necessary for the SMT encoding.

\noindent\textbf{Variables.}
We start by fixing the set of variables used to represent the (unknown) model run and alignment:
\begin{compactenum}[(a)]
\item Transition variables $\transvar_j$ of type integer for all $1\leq j\leq n$ to identify 
the $j$-th transition in the run. To this end, we enumerate the transitions as $T = \{t_1, \dots, t_L\}$, and add the constraint $\bigwedge_{j=1}^n 1\leq \transvar_j \leq L$,
with the semantics that $\transvar_j$ is assigned value $l$ iff the $j$-th transition in $\rho$ is $t_l$.
\item To identify the markings in the run, we use marking variables $\markvar_{j,p,\vec o}$ of type boolean for every time point $0\leq j\leq n$, every place $p\in P$, and every vector $\vec o$ of objects with elements in $O$ such that $\coloring(\vec o)=\coloring(p)$. The semantics is that $\markvar_{j,p,\vec o}$ is assigned true iff $\vec o$ occurs in $p$ at time $j$.
\item To keep track of which objects are used by transitions of the run, we use object variables $\objvar_{j,k}$ of type integer for all $1\leq j\leq n$ and $0\leq k \leq K$ with the constraint $\bigwedge_{j=1}^n 1\leq \objvar_{j,k} \leq |O|$. 
The semantics is that if $\objvar_{j,k}$ is assigned value $i$ then, if $i>0$ the $k$-th object involved in the $j$-th transition is $o_i$, and if $i=0$ then the $j$-th transition uses less than $k$ objects.
\end{compactenum}
In addition, we use the following variables to represent alignment cost:
\begin{compactenum}
\item[(d)] Distance variables $\distvar_{i,j}$ of type integer for every $0\leq i\leq m$ and $0\leq j\leq n$, their use will be explained later.
\end{compactenum}
\smallskip

\noindent\textbf{Constraints.}
We use the following constraints on the variables defined above:
\begin{compactenum}[(1)]
\item 
\emph{Initial markings}.
We first need to ensure that the first marking in the run $\rho$ is initial.
By the expression $[\vec o \in M(p)]$ we abbreviate $\top$ if an object tuple $\vec o$ occurs in the $M(p)$, and $\bot$ otherwise.
\begin{align}
\label{eq:phi:initial}
\tag{$\varphi_{\mathit{init}}$}
\textstyle
\bigvee_{M \in M_{init}}
\bigwedge_{p\in P} \bigwedge_{\vec o \in \vec O_{\coloring(p)}} \markvar_{0,p, \vec o} = [\vec o \in M(p)] 
\end{align}
\item 
\emph{Final markings.} 
Next, we state that after at most $n$ steps, but possibly earlier, a final marking is reached.
\begin{align}
\label{eq:phi:final}
\tag{$\varphi_{\mathit{fin}}$}
\bigvee_{0 \leq j \leq n}
\bigvee_{M \in M_{final}}
\textstyle\bigwedge_{p\in P} \bigwedge_{\vec o \in \vec O_{\coloring(p)}} \markvar_{j,p, \vec o} = [\vec o \in M(p)] 
\end{align}
 \item 
\emph{Moving tokens.}
Transitions must be enabled, and tokens are moved by transitions. We encode this as follows:
\begin{align}
\bigwedge_{j=1}^n \bigwedge_{l=1}^L \transvar_j=l \rightarrow 
&\bigwedge_{p \in \pre{t_l}\setminus \post{t_l}} \bigwedge_{\vec o\in \vec O_{\coloring(p)}} (\constoken(p,t_l,j,\vec o) \to  \markvar_{j-1,p,\vec o} \wedge \neg\markvar_{j,p,\vec o}) \wedge {}
\notag\\
&\bigwedge_{p \in \pre{t_l}\cap \post{t_l}} \bigwedge_{\vec o\in \vec O_{\coloring(p)}} (\constoken(p,t_l,j,\vec o) \to  \markvar_{j-1,p,\vec o}) \wedge {}
\notag \\
&\bigwedge_{p \in \post{t_l}} \bigwedge_{\vec o\in \vec O_{\coloring(p)}} 
(\prodtoken(p,t_l,j,\vec o) \to  \markvar_{j,p,\vec o})
\wedge {}
\notag \\
&\synced(p,t_l,j)
\tag{$\varphi_{\mathit{move}}$}
\label{eq:phi:tokens}
\end{align}
where $\constoken(p,t,j,\vec o)$ expresses that token $\vec o$ is consumed from $p$ in the $j$th transition which is $t$,  similarly
$\prodtoken(p,t,j,\vec o)$ expresses that token $\vec o$ is produced, and $\synced(p,t_l,j)$ ensures that, in the case where the flow from $p$ to $t$ uses an $=$-template inscription, all tokens in $p$ are consumed.
Formally, $\constoken$ is encoded as follows, distinguishing two cases: 
\begin{compactitem}
\item 
if $\inflow(p,t) = (v_1, \dots, v_h)$ is a non-variable flow, let
$(k_1, \dots, k_h)$ be the object indices for $t$ of $v_1, \dots, v_h$. 
Then
\[\constoken(p,t,j,\vec o) := (\bigwedge_{i=1}^{h} \objvar_{j,k_i} = id(\vec o_i))\] 
i.e., we demand that every variable used in the transition is instantiated to the respective object in $\vec o$.
In this case, we set $\synced(p,t_l,j) =\top$.
\item
if $\inflow(p,t) = (V_1, \dots, v_h)$ is a variable flow, suppose without loss of generality that $V_1\in \listvarset$.
Variable $V_1$ can be instantiated by multiple objects in a transition firing.
This is also reflected by the fact that there are several (but at most $K$) inscription indices
corresponding to instantiations of $V_1$, say $\ell_1, \dots, \ell_x$.
For $k_i$ as above for $i>1$, we then set
\[\constoken(p,t,j,\vec o) := (\bigwedge_{i=2}^{h} \objvar_{j,k_i} = id(\vec o_i)) \wedge \bigvee_{i=1}^x \objvar_{j,\ell_i} = id(\vec o_1)\]
If $V_1$ is a $\subseteq$-template inscription, we set again
$\synced(p,t_l,j) =\top$. Otherwise, $V_1$ is a $=$-template inscription, and it must be ensured that all tokens from $p$ are consumed. To this end, we set
\[
\synced(p,t_l,j) = 
\sum_{\vec o\in \vec O_{\coloring(p)}} \markvar_{j-1,p,\vec o} = \sum_{i=1}^h (\objvar_{j,k_i} \neq 0)
\]
i.e., the number of tokens present in $p$ at instant $j-1$ must be equal to the number of objects used to instantiate $V_1$. (Note that, formally, sums over boolean expressions must be encoded using if-then-else constructs; they are omitted here for readability.)
\end{compactitem}
The shorthand $\prodtoken$ is encoded similarly as $\constoken$, using $\outflow(t,p)$.
 \item 
\emph{Tokens that are not moved by transitions remain in their place.}
\begin{align}
\notag
\bigwedge_{j=1}^{n+1} 
\bigwedge_{p \in P}
\bigwedge_{\vec o\in \vec O_{\coloring(p)}} (\markvar_{j-1,p,\vec o} \leftrightarrow \markvar_{j,p,\vec o}) \vee &\bigvee_{t_l \in \post{p}} (\transvar_j\eqn l \wedge \constoken(p,t,j,\vec o) ) \vee {}\\
\label{eq:phi:inertia}
\tag{$\varphi_{\mathit{rem}}$}
& \bigvee_{t_l \in \pre{p}} (\transvar_j\eqn l \wedge \prodtoken(p,t,j,\vec o) ) 
\end{align}
 \item 
\emph{Transitions use objects of suitable type.}
To this end, recall that every transition can use at most $K$ objects, which limits instantiations of template inscriptions.
For every transition $t\in T$, we can thus enumerate the objects used by it from 1 to $K$.
However, some of these objects may be unused. We use the shorthand $\mathit{needed}_{t,k}$ to express this: $\mathit{needed}_{t,k} = \top$ if the $k$-th object is necessary for transition $t$ because it occurs in a simple inscription, and $\bot$ otherwise.
Moreover, let $\mathit{ttype}(t,k)$ be the type of the $k$-th object used by transition $t$.
Finally, we denote by $O_\sigma$ the subset of objects in $O$ of type $\sigma$.
\begin{equation}
\label{eq:phi:type}
\tag{$\varphi_{\mathit{type}}$}
\bigwedge_{j=1}^n \bigwedge_{l=1}^L \transvar_j=l \rightarrow \bigwedge_{k=1}^K \left ((\neg [\mathit{needed}_{t_l,k}] \wedge \objvar_{j,k}=0) \vee \bigvee_{o \in O_{\mathit{ttype}(t_l,k)}} \objvar_{j,k}=id(o)\right)
\end{equation}
 \item 
\emph{Objects that instantiate $\nu$-variables are fresh.}
We assume in the following constraint that $tids_\nu$ is the set of all $1 \leq l \leq L$ such that $t_l$ has an outgoing $\nu$-inscription, and that every such $t_l$ has only one outgoing $\nu$-inscription $\nu_t$, and we assume w.l.o.g. that in the enumeration of objects of $t$, $\nu_t$ is the first object. However, the constraint can be easily generalized to more such inscriptions.
\begin{equation}
\label{eq:phi:fresh}
\tag{$\varphi_{\mathit{fresh}}$}
\bigwedge_{j=1}^n \bigwedge_{l\in tids_\nu} \bigwedge_{o\in O_{\otype(\nu_t)}} \transvar_j=l \wedge \objvar_{j,1} = id(o)  \rightarrow 
(\bigwedge_{p\in P} \bigwedge_{\vec o \in \vec O_{\coloring(p)}, o \in \vec o} \neg \markvar_{j-1,p,\vec o})
\end{equation}
\item 
\textit{Guards are satisfied.} To that end, we set
\begin{equation}
\label{eq:phi:guard}
\tag{$\varphi_{\mathit{guard}}$}
\bigwedge_{j=1}^n \bigwedge_{l=1}^L \transvar_j=l \rightarrow
\guass(t_l)(\objvar_{j,1}, \dots, \objvar_{j,K})
\end{equation}
where $\guass(t_l)(\objvar_{j,1}, \dots, \objvar_{j,K})$ is an instantiation of the guard of $t_l$ with the object variables of instant $j$, using object indices. Here we assume that aggregation functions have a suitable SMT encoding supported by the solver, which is the case for the common aggregations of summation, maximum, minimum, and average.
\end{compactenum}

\noindent\textbf{Encoding alignment cost.}
Similar as in~\cite{FelliGMRW23,BoltenhagenCC21}, we encode the cost of an alignment as the edit distance with respect to suitable penalty functions $P_=$, $P_M$, and $P_L$.
Given a trace graph
$T_X = (E_X, D_X)$, let
\begin{equation}
\logtrace = \tup{e_1, \dots, e_m}
\label{eq:ordered:trace}
\end{equation}
be an enumeration of all events in $E_X$ such that
$\projtime(e_1) \leq \dots \leq \projtime(e_m)$.
Let the penalty expressions $[P_L]_i$, $[P_M]_{j}$, and $[P_=]_{i,j}$
be as follows, for all $1\leq i \leq m$ and $1 \leq j \leq n$:
\begin{align*}
[P_L]_{i} &= |\projobj(e_i)| \qquad\qquad [P_=]_{i,j} = 
ite(\mathit{is\_labelled}(j, \projact(e_i)), 0, \infty)\\
[P_M]_{j} &=
ite(\mathit{is\_labelled}(j, \tau)
, 0, \Sigma_{k=1}^K (\objvar_{j,k} \neq 0) )
\end{align*}
where $\mathit{is\_labelled}(j,a)$ expresses that the $j$-the transition has label $a\in \activities \cup \{\tau\}$, which can be done by taking $is\_labelled(j,a):= \bigvee_{l\in T_{idx}(a)}\transvar_j = l$
where $T_{idx}(a)$ is the set of transition indices with label $a$, i.e., the set of all $l$ with $t_l \in T$ such that $\ell(t_l) = a$.

Using these expressions, one can encode the edit distance as in \cite{FelliGMRW23,BoltenhagenCC21}:
\begin{equation}
 \label{eq:delta}
\begin{array}{rl@{\qquad}rl@{\qquad}rl@{\qquad\quad}r}
\distvar_{0,0} &= 0 &
\distvar_{{i+1},0} &= [P_L] + \distvar_{i,0} &
\distvar_{0,{j+1}} &= [P_M]_{j+1} + \distvar_{0,j} 
\\[1ex]
\distvar_{i+1,j+1} &\multicolumn{5}{l}{=
\min (
[P_=]_{i+1, j+1} + \distvar_{i,j},\ 
[P_L] + \distvar_{i,j+1},\ 
[P_M]_{j+1} + \distvar_{i+1,j})}
\end{array}
\tag{$\varphi_\delta$}
\end{equation}

\noindent\textbf{Solving.}
We abbreviate 
$\varphi_{\mathit{run}} = \varphi_{\mathit{init}} \wedge \varphi_{\mathit{fin}} \wedge \varphi_{\mathit{move}} \wedge \varphi_{\mathit{rem}} \wedge \varphi_{\mathit{type}} \wedge \varphi_{\mathit{fresh}}\wedge \varphi_{\mathit{guard}}$ and
use an SMT solver to obtain a satisfying assignment $\alpha$ for the 
following constrained optimization problem: 

\begin{align*}
\label{eq:constraints}
\varphi_{\mathit{run}} \wedge
\varphi_{\delta}
\text{\quad minimizing\quad }\distvar_{m,n}
\tag{$\Phi$}
\end{align*}

\noindent\textbf{Decoding.}
From an assignment $\alpha$ satisfying \eqref{eq:constraints}, we next define a run $\rho_\alpha$ and an alignment $\Gamma_\alpha$.
First, we note the following:
From Lemma~\ref{lem:bounds:precise}, we can obtain a number $M$ such that $M$ is the maximal number of objects used to instantiate a list variable in the model run and alignment. By convention, we may assume that in the enumeration of objects used in the $j$th transition firing, $\objvar_{j, |O|-M+1}, \dots, \objvar_{j, |O|}$ are those instantiating a list variable, if there is a list variable in $\invars{t_{\alpha(\transvar_j)}} \cup \outvars{t_{\alpha(\transvar_j)}}$.

We assume the set of transitions $T=\set{t_1, \dots, t_L}$ is ordered as $t_1, \dots, t_L$ in some arbitrary but fixed way that was already used for the encoding. 

\begin{definition}[Decoded run]
\label{def:decoded:run}
For $\alpha$ satisfying \eqref{eq:constraints}, let the \emph{decoded process run} be
$\rho_\alpha =  \tup{f_1, \dots, f_n}$ such that for all $1\leq j \leq n$,  $f_j = (\widehat t_j, b_j)$, where $\widehat t_j  = t_{\alpha(\transvar_j)}$ and $b_j$ is defined as follows:
Assuming that $\invars{t_{\alpha(\transvar_j)}} \cup \outvars{t_{\alpha(\transvar_j)}}$ is ordered as $v_1, \dots, v_k$ in an arbitrary but fixed way that was already considered for the encoding,
we set $b_j(v_i) = \alpha(\objvar_{j,i})$ if $v_i \in \varset$, and 
$b_j(v_i) = [O_{\alpha(\objvar_{j,|O|-M+1})}, \dots, O_{\alpha(\objvar_{j,|O|-M+z})}]$ if $v_i \in \varset$, where $0 \leq z < M$ is maximal such that $\alpha(\objvar_{j,|O|-M+z}) \neq 0$.
\end{definition}

At this point, $\rho_\alpha$ is actually just a sequence; we will show below that it is indeed a process run of $\NN$.
Next, given  a satisfying assignment $\alpha$ for \eqref{eq:constraints}, we define an alignment of the log trace $T_X$ and the process run $\rho_\alpha$.

\begin{definition}[Decoded alignment]
\label{def:decoded:alignment}
For $\alpha$ satisfying \eqref{eq:constraints}, $\rho_\alpha =  \tup{f_1, \dots, f_n}$ as defined above, and $\logtrace$ as in \eqref{eq:ordered:trace},
consider the sequence of moves $\gamma_{i,j}$ recursively defined as follows:
\begin{align*}
\gamma_{0,0} &=\epsilon \qquad
\gamma_{i+1,0}= \gamma_{i,0} \cdot \tup{e_{i+1}, \SKIP} \qquad
\gamma_{0,j+1}= \gamma_{0,j} \cdot \tup{\SKIP, f_{j+1}} \\
\gamma_{i+1,j+1} &= 
\begin{cases}
\gamma_{i,j+1} \cdot \tup{e_{i+1}, \SKIP} &
 \text{ if }\alpha(\distvar_{i+1,j+1}) = \alpha([P_L] + \distvar_{i,j+1}) \\
\gamma_{i+1,j} \cdot \tup{\SKIP, f_{j+1}} &
 \text{ if otherwise }\alpha(\distvar_{i+1,j+1}) = \alpha([P_M]_{j+1} + \distvar_{i+1,j}) \\
\gamma_{i,j} \cdot \tup{e_{i+1}, f_{j+1}} &
 \text{ otherwise}
\end{cases}
\end{align*}
Given $\gamma_{i,j}$, we define a 
graph $\Gamma(\gamma_{i,j})=\tup{C,B}$ of moves as follows: the node set $C$ consists of all moves in $\gamma_{i,j}$, and there is an edge $\tup{\tup{q,r}, \tup{q',r'}} \in B$ if either $q\neq \SKIP$, $q' \neq \SKIP$ and there is an edge $q \to q'$ in $T_X$, or
if $r\neq \SKIP$, $r' \neq \SKIP$, $r=f_h$, and $r'=f_{h+1}$ for some $h$ with $1 \leq h < n$.
Finally, we define the decoded alignment as $\Gamma(\alpha) := \Gamma(\gamma_{m,n})$.
\end{definition}
In fact, as defined, $\Gamma(\alpha)$ is just a graph of moves, it yet has to be shown that it is a proper alignment. This will be done in the next section.
\smallskip

\noindent\textbf{Correctness.}
In the remainder of this section, we will prove that $\rho_\alpha$ is indeed a run, and $\Gamma(\alpha)$ is an alignment of $T_X$ and $\rho_\alpha$. We first show the former:

\begin{lemma}
\label{lem:decode:run}
Let $\NN$ be \anet, $T_X$ a log trace and $\alpha$ a solution to $\eqref{eq:constraints}$. Then $\rho_\alpha$ is a run of $\NN$.
\end{lemma}
\begin{proof}
We define a sequence of markings $M_0, \dots, M_n$.
Let $M_j$, $0 \leq j \leq n$, be the marking such that $M_j(p)= \{\vec o \mid \vec o \in \vec O_{\coloring(p)} \text{ and } \alpha(\markvar_{j,p,\vec o})=\top\}$.
Then, we can show by induction on $j$ that for the process run $\rho_j = \tup{f_1, \dots, f_j}$ 
it holds that  $M_0\goto{\rho_j} M_j$.
\begin{itemize}
\item[\textit{Base case.}] If $n=0$, then $\rho_0$ is empty, so the statement is trivial. 
\item[\textit{Inductive step.}] Consider $\rho_{j+1} = \tup{f_1, \dots, f_{j+1}}$ and
suppose that for the prefix
$\rho' = \tup{f_1, \dots, f_j}$
it holds that $M_0 \goto{\rho'} M_j$.
We have $f_{j+1} = (\widehat t,b)$ and $\widehat t=t_i$ for some $i$ such that $1\,{\leq}\,i\,{\leq}\,|T|$ with $\alpha(\transvar_j) = i$.
First, we note that $b$ is a valid binding: as $\alpha$ satisfies \eqref{eq:phi:type}, it assigns a non-zero value to all $\objvar_{j,k}$ such that $v_k \in \invars{t_i} \cup \outvars{t_i}$ that are not of list type (and hence needed), and by $\eqref{eq:phi:type}$, the unique object $o$ with $id(o)=\alpha(\objvar_{j,k})$ has the type of $v_k$.
Similarly, $b$ assigns a list of objects of correct type to a variable in
$(\invars{t_i} \cup \outvars{t_i}) \cap \listvarset$, if such a variable exists.
Moreover, \eqref{eq:phi:fresh} ensures that variables in $(\invars{t_i} \cup \outvars{t_i}) \cap \nuvarset$ are instantiated with objects that did not occur in $M_j$,
and \eqref{eq:phi:guard} ensures taht the guard is satisfied.

Since $\alpha$ is a solution to $\eqref{eq:constraints}$, it satisfies \eqref{eq:phi:tokens}, so that $t_i$ is enabled in $M_n$.
Moreover, the distinction between $\subseteq$- and $=$-inscriptions is taken care of by the $\synced$ predicate.
As $\alpha$ satisfies \eqref{eq:phi:inertia}, the new marking $M_{j+1}$ contains only either tokens that were produced by $t_i$, or tokens that were not affected by $t_i$. Thus, $M_j\goto{f_{j+1}} M_{j+1}$, which concludes the induction proof.

\end{itemize}
Finally, as $\alpha$ satisfies 
\eqref{eq:phi:initial} and \eqref{eq:phi:final}, it must be that $M_0=M_I$ and the last marking must be final, so $\rho_\alpha$ is a run of $\NN$.
\qed
\end{proof}

\begin{theorem}
Given \anet $\NN$, trace graph $\tracenet_X$, and satisfying assignment $\alpha$ to $\eqref{eq:constraints}$, $\Gamma(\alpha)$ is an optimal alignment of $\tracenet_X$ and the run $\rho_\alpha$  with cost $\alpha(\distvar_{m,n})$.
\end{theorem}
\begin{proof}
By \lemref{decode:run}, $\rho_\alpha$ is a run of $\NN$.
We first note that $[P_=]$, $[P_L]$, and $[P_M]$ are correct encodings of 
$P_=$, $P_L$, and $P_M$ from \defref{cost}, respectively. For $P_L$ this is clear.
For $P_=$, $\mathit{is\_labelled}(j,a)$ is true iff the value of $\transvar_j$ corresponds to a transition
that is labeled $a$. If the labels match, cost $0$ is returned, otherwise $\infty$.
For $P_M$, the case distinction returns cost $0$ if the $j$th transition is silent; otherwise, the expression 
$\Sigma_{k=1}^K ite(\objvar_{j,k} \neq 0, 1, 0)$ 
counts the number of objects involved in the model step, using the convention that if fewer than $k$ objects are involved in the $j$th transition then $\objvar_{j,k}$ is assigned 0.

Now, let $d_{i,j} = \alpha(\distvar_{i,j})$, for all $i$, $j$ such that $0 \leq i \leq m$ and $0 \leq j \leq n$.
Let again $\logtrace=\tup{e_1, \dots, e_m}$ be the sequence ordering the nodes in $T_X$ as in \eqref{eq:ordered:trace}. 
Let $T_X|_i$ be the restriction of $T_X$ to the node set $\{e_1, \dots, e_i\}$.
We show the stronger statement that $\Gamma(\gamma_{i,j})$ is an optimal alignment of $T_X|_i$ and $\rho_\alpha|_j$ with cost $d_{i,j}$, by induction on $(i,j)$. 

\begin{itemize}
\item[\textit{Base case.}] If $i\,{=}\,j\,{=}\,0$, 
then $\gamma_{i,j}$ is the trivial, empty alignment of an empty log trace and an empty process run, which is clearly optimal with cost
$d_{i,j}\,{=}\,0$, as defined in \eqref{eq:delta}.
\item[\textit{Step case.}]
If $i\,{=}\,0$ and $j\,{>}\,0$, then $\gamma_{0,j}$ is a sequence of model moves 
$\gamma_{0,j}=\tup{(\SKIP, f_{1}),\ldots, (\SKIP, f_{j}) }$ according to \defref{decoded:alignment}.
Consequently, $\Gamma(\alpha)=\Gamma(\gamma_{0,j})$ has edges $(\SKIP, f_{h}),\ldots, (\SKIP, f_{h+1})$ for all $h$, $1\leq h < j$, which is a valid and optimal alignment of the empty log trace and $\rho_\alpha$.
By \defref{cost}, the cost of $\Gamma(\alpha)$ is the number of objects involved in 
non-silent transitions of $f_1, \dots, f_j$, which coincides with
$\alpha([P_M]_1 + \dots + [P_M]_j)$, as stipulated in \eqref{eq:delta}.
\item[\textit{Step case.}]
If $j\,{=}\,0$ and $i\,{>}\,0$, then $\gamma_{i,0}$ is a sequence of log moves 
$\gamma_{i,0}=\tup{(e_{1},\SKIP),\ldots, (e_{j}, \SKIP) }$ according to \defref{decoded:alignment}.
Thus, $\Gamma(\alpha)=\Gamma(\gamma_{i,0})$ is a graph whose log projection
coincides by definition with $T_X|_i$.
By \defref{cost}, the cost of $\Gamma(\alpha)$ is the number of objects involved in 
$e_1, \dots, e_i$, which coincides with
$\alpha([P_L]_1 + \dots + [P_L]_j)$, as stipulated in \eqref{eq:delta}.
\item[\textit{Step case.}] 
If $i\,{>}\,0$ and  $j\,{>}\,0$, 
$d_{i,j}$ must be the minimum of
$\alpha([P_=]_{i, j}){+}d_{i-1,j-1}$,
$\alpha([P_L]){+}d_{i-1,j}$, 
and $\alpha([P_M]_{j}){+}d_{i,j-1}$.
We can distinguish three cases:
\begin{compactitem}
\item
Suppose $d_{i,j} = \alpha([P_L]_i) + d_{i-1,j}$.
By \defref{decoded:alignment},
we have $\gamma_{i,j} = \gamma_{i-1,j} \cdot \tup{e_{i}, \SKIP}$.
Thus, $\Gamma(\gamma_{i,j})$ extends $\Gamma(\gamma_{i-1,j})$ by a node
$\tup{e_{i}, \SKIP}$, and edges to this node as induced by $T_X|_i$.
By the induction hypothesis, $\Gamma(\gamma_{i-1,j})$ is a valid and optimal alignment
of $T_X|_{i-1}$ and $\rho_\alpha|_{j}$ with cost $d_{i-1,j}$.
Thus $\Gamma(\gamma_{i,j})$ is a valid alignment of $T_X|_i$ and $\rho_\alpha|_{j}$,
because the log projection coincides with $T_X|_i$ by definition.
By minimality of the definition of $d_{i,j}$, also $\Gamma(\gamma_{i,j})$ is optimal.
\item
Suppose $d_{i,j} = \alpha([P_M]_{j}) + d_{i,j-1}$.
By \defref{decoded:alignment},
we have $\gamma_{i,j} = \gamma_{i,j-1} \cdot \tup{\SKIP, f_j}$.
By the induction hypothesis, $\Gamma(\gamma_{i,j-1})$ is a valid and optimal alignment
of $T_X|_{i}$ and $\rho_\alpha|_{j-1}$ with cost $d_{i,j-1}$.
Thus, $\Gamma(\gamma_{i,j-1})$ must have a node $\tup{r,f_{j-1}}$, for some $r$.
The graph $\Gamma(\gamma_{i,j})$ extends $\Gamma(\gamma_{i,j-1})$ by a node
$\tup{\SKIP, f_{j}}$, and an edge $\tup{r,f_{j-1}} \to \tup{\SKIP, f_{j}}$.
Thus, $\Gamma(\gamma_{i,j})$ is a valid alignment for $T_X|_{i}$ and $\rho_\alpha|_{j}$, and by minimality it is also optimal.
\item
Let $d_{i,j} = \alpha([P_=]_{i, j}) + d_{i-1,j-1}$.
By \defref{decoded:alignment},
we have $\gamma_{i,j} = \gamma_{i-1,j-1} \cdot \tup{e_i, f_j}$.
By the induction hypothesis, $\Gamma(\gamma_{i-1,j-1})$ is an optimal alignment
of $T_X|_{i-1}$ and $\rho_\alpha|_{j-1}$ with cost $d_{i-1,j-1}$.
In particular, $\Gamma(\gamma_{i-1,j-1})$ must have a node $\tup{r,f_{j-1}}$, for some $r$. The graph $\Gamma(\gamma_{i,j})$ extends $\Gamma(\gamma_{i-1,j-1})$ by a node
$\tup{e_1, f_{j}}$, an edge $\tup{r,f_{j-1}} \to \tup{e_i, f_{j}}$, as well as edges to $\tup{e_1, f_{j}}$ as induced by $T_X|_i$.
Thus $\Gamma(\gamma_{i,j})$ is a valid alignment of $T_X|_i$ and $\rho_\alpha|_{j}$,
because the log projection coincides with $T_X|_i$ by definition, and the model projection has the required additional edge.
By minimality of the definition of $d_{i,j}$, also $\Gamma(\gamma_{i,j})$ is optimal.
\end{compactitem}
%
\end{itemize}
For the case $i=m$ and $j=n$, we obtain that $\Gamma(\alpha)=\Gamma(\gamma_{m,n})$ is an optimal alignment of $\tracenet_X$ and $\rho_\alpha$ with cost $d_{m,n}=\alpha(\distvar_{m,n})$.
\qed
\end{proof}

}

\end{document}